\documentclass[11pt]{article}

\usepackage{mdwtab}

\usepackage[cmex10]{amsmath}
\usepackage{amsfonts,amssymb,amsthm}
\usepackage{booktabs}
\usepackage[ruled,lined]{algorithm2e}

\usepackage{pifont}
\usepackage{geometry}

\newtheorem{theorem}{Theorem}[section]

\newtheorem{corollary}[theorem]{Corollary}
\newtheorem{proposition}[theorem]{Proposition}

\newtheorem{lemma}[theorem]{Lemma}

\newtheorem{remark}[theorem]{Remark}

\numberwithin{equation}{section}

\newcommand{\eat}[1]{}

\begin{document}

\title{Energy-Efficient Scheduling: \\Classification, Bounds, and
  Algorithms}
\author{Pragati Agrawal \hspace{0.7in} Shrisha Rao \\%
        {\tt pragati.a.in@ieee.org} \hspace{0.2in} {\tt shrao@ieee.org} }

\date{}
\maketitle

\section{Introduction}

Energy is a precious resource of the current industrial economy.  For
energy conservation and sustainability, it is necessary to achieve
energy efficiency, which we may take to mean minimization of the
energy consumed to do a given amount of work.  This requires ways to
achieve greater energy efficiencies by proper scheduling of jobs over
machines, in addition to making each machine individually more energy
efficient.  Scheduling theory is an important field of study that has
received attention from researchers for decades, most work on
scheduling has looked at makespan and other time-related
objectives.  Energy being such a limited resource and directly related
to cost, it can in certain instances be considered more important than
time, i.e., it may be better to complete jobs using minimum energy,
rather than completing them quickly. Clearly, scheduling jobs in such
a way that the total energy consumption of the system is minimal, or
at least low, is of the essence.  We present a generalized theory for
energy-efficient offline scheduling in this paper, which does not go
into domain-specific issues or technologies.

Our work is applicable to systems whose machines are similar in their
capabilities but can have different working and idle power
consumptions.  The speeds at which machines can execute tasks can be
different from one another, but are not time-varying, i.e., each
machine operates at the same speed whenever working.  All machines are
connected to one another, so that a job from one machine can be
transferred to any other machine.  There is no set limit to the number
of machines in a system.  Machines and jobs to be executed are
independent, and hence any job can be executed in any order on any
machine.  Each machine is characterized by a fixed working power (the
power consumed by it when under load) and a fixed idle power (the
power consumed by it when running idle), as seen in previous
work~\cite{pagrawal2014}. Given a set of interconnected machines and
independent jobs, our results specify the allocation of jobs on the
machines so that the total power consumption of the system is
minimized.  However, unlike that previous work, we focus on exact
results rather than using heuristics for scheduling.

Considering a system of machines cooperatively running similar jobs,
we classify the problems of energy-minimal offline scheduling with
non-identical interconnected machines and independent jobs as follows:

\begin{enumerate}

 \item \emph{Identical speeds, divisible jobs}:\\ This class of
   problem is discussed in Section~\ref{subsub:id_div}, where we
   derive results to prove that energy-minimal scheduling under this
   type of system can be solved in linear time; we also give an
   algorithm for the same.

 \item \emph{Identical speeds, non-divisible jobs}:\\ This turns out
   to be NP-hard, as shown in Section~\ref{subsub:id_nd}.  We also
   give a linear-time approximation algorithm which gives a solution
   within a bound of $\frac{4}{3}$.

 \item \emph{Different speeds, divisible jobs}:\\ In
   Section~\ref{subsub:diff_div} we give results for energy-minimal
   scheduling on such a system and also a linear-time algorithm to
   find the minimum-energy schedule.

 \item \emph{Different speeds, non-divisible jobs}:\\ As with the
   case of identical speed machines, this class of problems is also
   NP-hard.  We give a linear-time approximation algorithm for this
   class in Section~\ref{subsub:diff_nd} which gives a result within
   the bound of $1+\sqrt{3}/3 \approx 1.5773$ times the optimal.

\end{enumerate}

We first introduce a generic system model which abstracts the relevant
aspects of energy-efficient scheduling.  We then give a classification
of the problems of energy-efficient scheduling on different types of
systems, and study their complexities.  This makes possible analyses
giving results which govern the relative distribution of work among
machines for maximum energy efficiency, given the energy
specifications of the machines.  Using these results we design
scheduling algorithms for different classes of systems.  For problem
classes in which scheduling is NP-hard, we give approximation
algorithms along with their bounds.

We deliberately do not give specific units for energy (which is
typically measured in joules, ergs, kilowatt-hours, etc.), power
(often measured in watts, etc.), work (generally specified in the same
units as energy), or speed (work done per unit time); we likewise do
not find it necessary to be precise about types of energy or energy
sources.  This is so to remain agnostic towards particular system
technologies and instantiations, and to avoid domain-specific biases
in our analyses.

There is a lot of work done for makespan scheduling, and other
time-related objectives.  But solving for time-related objectives does
not guarantee minimum energy.  The novelty of this work is in
considering energy specifically with completely generic machine
specifications.  \eat{ No other theoretical work until now has considered
energy consumption of machines while idle, which is known to be
significant from real life; e.g., in data
centers~\cite{datacenter1,datacenter2}. } Our system model also shows
the absolute bounds of energy-minimal scheduling for specific types of
systems.  We derive results (for divisible jobs) which serve as an
upper-bound to maximum achievable efficiency---because a system with
divisible jobs can always have higher efficiency as compared to system
with non-divisible jobs~\cite{pinedo2012scheduling}.

A part of this theoretical framework, seen here in
Section~\ref{subsec:i_sp}, has been presented in our workshop
paper~\cite{pagrawal2015}, which introduces the classification of
energy-minimal scheduling problems and gives results on the class of
systems which have identical speed machines and divisible jobs.  It
also gives results for a special case of system with different speed
machines and divisible jobs in which idle power consumption of
machines was considered to be zero.  It can be seen as the beginning of
the work presented in the current paper which comprehensively covers
most types of energy-minimal scheduling problems with completely
interconnected machines and independent jobs.

In the present paper we deal with theoretical aspects of
energy-minimal scheduling, deriving results which can guide the design
of scheduling algorithms for systems.  On any given schedule, if the
conditions indicated in the results are satisfied, then that schedule
will be energy-minimal.  Likewise, the results can also indicate when a
system of machines would be inherently wasteful of energy, and can
thus guide more energy-efficient system designs.

It is easy to see (even a back-of-the-envelope calculation can
suffice) that there are many cases where it is not appropriate to run
all machines, if one's goal is to complete a set of jobs with minimal
energy consumption.  The question that arises is: given a set of jobs
to complete, which machines should one use, and which not?  Our
results and analyses address this point.  

Machines not in use are not assumed to be switched off, but do consume
power (what we call idle power).  We have deliberately allowed that
machines do not get switched off when not in use, given the reality of
machines in many domains (after all, the on-off cycle time for an
industrial furnace, or even a personal computer, is not small).  The
situation where a machine does get switched off is merely a special
case where the idle power of that machine is zero.

For the class of system in which machines have identical speeds with
different power ratings, and when the jobs are divisible, we first
determine the precedence of the machines (the order in which they
should be assigned work).  We then derive the optimal number of
machines which should be working so that the energy consumption is
minimum.  Based on this we derive what should be the optimal
allocation of work among the machines chosen to run.  We give an
$\mathcal{O}(m)$-time approximation (where $m$ is the number of machines) 
algorithm which 
instantiates the theoretical
results, to derive the energy-minimal schedule given the
specifications of the machines and jobs.

The slightly harder problem of finding the minimum-energy schedule for
divisible jobs over machines with different speeds and power ratings
can also be solved in $\mathcal{O}(m)$ time.  We give an exact
algorithm for this class too. 

For the classes of problems with non-divisible jobs, we see that the
problems are NP-hard.  Hence we give $\mathcal{O}(m)$-time approximation
algorithms for both cases: systems with identical- and different-speed
machines.

The remainder of the paper is organized as follows. Related work is mentioned 
in Section~\ref{sec:rw}. The system model along with the notation used in 
defining the energy-minimal scheduling problem and the motivation is presented 
in Section~\ref{sec:pf}. Section~\ref{sec:method} discusses the 
precedence of machines for both the identical speed and the different speed 
machine 
case.  The proposed algorithms along with their bounds are given in 
Section~\ref{sec:anc}.  Finally, Section~\ref{sec:im} describes the two types 
of measures commonly used for energy efficiency in systems, and shows that 
these are incompatible.

\section{Related Work}
\label{sec:rw}

In energy-minimal scheduling, as opposed to classical makespan
scheduling, there are additional parameters (energy specifications)
which come into picture and hence it is inherently more complex than
makespan scheduling. Agrawal and Rao show that energy-efficient
scheduling is strictly harder than makespan scheduling, and present a
generic model for scheduling which takes into consideration both the
working power and the idle power of each machine~\cite{pagrawal2014}.  They 
propose three
heuristic algorithms for energy-aware scheduling in a system where
jobs have precedence constraints.

Though there has been a lot of work done in the field of scheduling,
the prime focus of general works on
scheduling~\cite{pinedo2012scheduling,Pinedo2009,Herrmann2006} has
been to optimize objectives related to time---such as makespan,
earliness, and avoidance of tardiness.

Existing literature dealing with scheduling for energy reduction
generally focuses on specific domains, such as communication networks,
embedded systems, and high performance computing.  These works
invariably rely on specific features and technologies of those
domains.  For instance,   Wang and Saksena       address parallel task
scheduling for reducing energy consumption for high-end computing
using the DVFS (dynamic voltage and frequency scheduling) technique
with some heuristics~\cite{wang1999}.
More recently, Yassa \emph{et al.} also base their
approach on DVFS to minimize energy consumption (as do many
others)~\cite{yassa2013}. Huang \emph{et al.} propose energy-aware task 
allocation using
simulated annealing with timing adjustment for network-on-chip based
heterogeneous multiprocessor systems~\cite{huang2011}. Sheikh \emph{et al.} survey
existing scheduling algorithms and software for reducing energy
dissipation in performing tasks on different platforms including
single processors, multicore processors, and distributed 
systems~\cite{Sheikh2012}.
This survey is quite comprehensive and discusses the architectural,
software, and algorithmic issues for energy-aware scheduling of
computing devices.  Pouwelse \emph{et al.} give a heuristic for energy
priority scheduling for variable voltage processors~\cite{Pouwelse01}.
Bambagini \emph{et al.} use the combination of offline-DVFS and online-DPM
(dynamic power management) techniques to reduce energy consumption in
embedded systems~\cite{Bambagini13}.  Some recent
anthologies deal with energy
efficiency in computing systems at various scales (processors to data
centers)~\cite{ZomayaLee2012,pierson2015}.  However, while they are rich in 
pragmatic approaches and
domain-specific ideas for energy efficiency, they do not present a
theory of energy-efficient scheduling \emph{per se}, and only basic
algorithms (e.g., Greedy-Min, Greedy-Max) are considered on occasion.
It is patently obvious that even though there are works on
application-specific energy-efficient scheduling, a generalized theory
is yet far from developed.

In the domain of microprocessor power management and energy saving, much of the 
prior work uses speed scaling and power-down 
techniques~\cite{irani2005,albers2009}. 
 Irani \emph{et al.} have given algorithms for power savings by putting the system on sleep state while idle and varying the speeds at which jobs are run~\cite{irani2007algorithms}.
Augustine \emph{et al.} suggest optimal power-down strategies with more than one low-power state~\cite{augustine2004}.  When a device is 
idle, their algorithm turns the device in low-power sleep state. 
Bansal \emph{et al.} use the other technique called speed 
scaling~\cite{bansal2009,bansal2013}. 
Bansal \emph{et al.} have given a tighter bound for the YDS algorithm given 
by Yao \emph{et al.} and proved that their BKP algorithm is cooling-oblivious~\cite{bansal2007}~\cite{yao1995} . 
In a paper Bansal \emph{et al.} consider an arbitrary power function of 
speed, for speed scaling~\cite{bansal2009,bansal2013}.

There are some studies of scheduling whose theoretical developments
can help in energy-efficient scheduling too.  One such is divisible
job theory, which studies methodologies involving the continuous and
linear modeling of partition-able communication and computation jobs
for parallel
processing~\cite{bharadwaj2003,yu2003,singh2014,abdullah2013,Robertazzi05}.
Another such field is multi-objective scheduling, in which machine
cost is considered as one of the objectives.  (In general, machine cost
includes any cost which is incurred because of usage of the machines of the
system.)

There is some prior work concerning scheduling for multiple criteria.
For instance, Leung \emph{et al.} consider two objectives for
scheduling: one is time-related such as makespan (as the customer's
objective), and the other is total machine cost (as the service
provider's objective)~\cite{leung2012}.  They form different final
objectives by combining these two objectives in various ways and
analyse the complexity of the scheduling problem for each of these
various combinations. Lee \emph{et al.} propose a heuristic for
bi-criteria scheduling with machine assignment costs with worst-case
performance bounds~\cite{lee2014}. Lee \emph{et al.} also work on
multi-objective scheduling\cite{lee2012}.  Other than makespan, they
also minimize the total congestion and completion time, the maximum
and total tardiness, and the number of tardy jobs.  They give a
coordination mechanisms for parallel machine scheduling.  Shi \emph{et
  al.} formulate an assignment scheme for the divisible job theory for
sensing workload allocations~\cite{shi2012}. Drozdowski \emph{et al.}
give a method of visualizing the relationships between computation and
energy consumption (in supercomputing and HPC applications) as
two-dimensional maps similar to isotherms~\cite{drozdowski2014}.

Though the objectives considered in these and other such works can be
seen as marginally relevant to minimizing energy (setting energy
consumed by a running machine as a machine cost), they have not
considered the idle power consumptions of machines, which is needed for
realistic analyses.  In many papers that consider multiple scheduling
objectives, all objectives are concerned with time only, and do not
make sense for energy-minimal scheduling.  The theoretical framework
and classification of energy-minimal scheduling problems is not
discussed in or derivable from any of these papers.

Clearly then, there is a need of a theoretical framework, as presented
here, which realistically models systems of machines considering both
idle power and running power, and also allows us to classify the
various types of problems involved in energy-minimal scheduling and
understand their complexities.  This theory can of course also be
extended for further fine-grained analyses (e.g., for machines with
multiple performance levels).  The theoretical framework proposed is
thus a necessary addition to existing and ongoing work focusing on
heuristic approaches to energy-efficient scheduling, and also
complements technology-specific approaches to energy reduction such as
DVFS~\cite{DVFS}.

\section{System Model} 
\label{sec:pf}

In this section we formally introduce our system model and define the
energy-minimal scheduling problem in the most general sense.  As any
other scheduling problem, we have a set of interconnected machines and
a set of jobs that are to be executed using the machines.  We presume
that all machines are connected in such a way that a job can be
transfered from any machine to any other machine without any energy
dissipation.  Jobs are independent and hence can be executed in any
order and on any machine.

We consider two types of jobs, divisible and non-divisible.  There can
be two types of systems as well, based on the speeds of the machines
in the system: one in which all machines run at the same speed (even
though their power ratings may be different), and the other class in
which machines run at different speeds.  Thus, based on the type of
job and the type of system, we have four classes of scheduling
problems.

Let the set of $m$ machines in a system be denoted by $\mathcal{C}$,
where $\mathcal{C} = \{c_i:1\leq i \leq m \}$ and $c_i$ denotes the
machine $i$ of the set $\mathcal{C}$.  We discuss the precedence of
machines later in this paper. $\mathcal{R}$ denotes the working set of machines.
The working power of machine $c_i$ is
denoted as $\mu(c_i)$, and the idle power as $\gamma(c_i)$.  The sum of
the idle power of all the machines is given by \(\sum_{i=1}^m
\gamma(c_i) = \Gamma\).  The speed of machine $c_i$ is denoted as
$\upsilon(c_i)$.  The speed (throughput of work per unit time) of a
machine is fixed throughout its working tenure.
 
All machines in a system work in parallel, and the maximum working
time of a system to execute a given set $\mathcal{P}$ of $n$ jobs
$p_1$ to $p_n$ is $T$, which is the makespan of the system for that
set of jobs.  If machine $c_i$ works only for time $\tau(c_i)$, then
the idle time of that machine $c_i$ is given by $T-\tau(c_i) =
\kappa(c_i)$.  The amount of work done by machine $c_i$ is represented
by $w(c_i)$, where
\begin{equation}
 w(c_i) = \tau(c_i) \upsilon(c_i)
\label{label:w}
\end{equation}

The weight (processing time) of job $p_j$ is denoted by $\psi(p_j)$ and the sum 
of the
weights of all jobs is given by $\sum_{j=1}^n \psi(p_j) = W$.  Also,
the sum of work done by all machines is equal to the total work to be
done, i.e., $\sum_{i=1}^m w(c_i) = W$.

When $\upsilon(c_i) = 1, \forall i, 1 \leq i \leq m$, then we say that
the machines are identical in their working capacities or speed,
implying that they can execute and complete any job given to them in
equal time.  Also, in this case $\tau(c_i) = w(c_i), \forall i, 1 \leq
i \leq m$.  And if in the present case all jobs are executed
sequentially on one machine, then that time taken is equal to $W$, the
total amount of work to be done.

The energy consumed by a machine can be calculated given the power
consumption by the machine in working state and the duration for which
the machine works.  Since the power rating of a machine is energy
consumed per unit time, it is in principle possible to calculate the
energy consumed by any machine over a given time duration by integrating
the power consumption over time~\cite{halliday2010}.  If the power
consumed is constant over time, the energy is given by the product of
power with time.

In many systems like data centers, it is not a simple matter to shut
off machines, even when they are not given any job.  Thus, if no
machine is switch off until the total work is over, even such an
idle machine consumes some power though not producing any work.
Considering the energy consumption in idle state too, the total energy
consumed by a machine is the sum of the energy consumed in the working
state and that consumed in the idle state.

The energy consumed by machine $c_i$ in the working state is given by
$\mu(c_i){\tau(c_i)}$, and the energy consumed in the idle state by
$\gamma(c_i)\kappa(c_i)$.  So the total energy consumption by a
machine $c_i$ is given by $\mu(c_i)\tau(c_i) +
\gamma(c_i)\kappa(c_i)$.  With this, the expression for the energy
consumed in a system can be given as the sum of the energies consumed
by all $m$ machines:

\begin{equation}
 E = \sum_{i = 1}^{m} [\mu(c_i)\tau(c_i) + \gamma(c_i)\kappa(c_i)]
\label{label:energyk}
\end{equation}
Putting the value of $\kappa(c_i) = T-\tau(c_i)$, in~\eqref{label:energyk},
\begin{equation}
 E = \sum_{i = 1}^{m} [\mu(c_i)\tau(c_i) + \gamma(c_i)(T-\tau(c_i))]
\label{label:energy}
\end{equation}
The general energy equation can be formed by replacing 
$\tau(c_i)$ with $\frac{w(c_i)}{\upsilon(c_i)}$ (from~\eqref{label:w}) 
in~\eqref{label:energy}.
\begin{equation*}
E = \sum_{i = 1}^{m} \left[\mu(c_i)\frac{w(c_i)}{\upsilon(c_i)} + 
\gamma(c_i)(T-\frac{w(c_i)}{\upsilon(c_i)})\right]
\end{equation*}
After simplification,
\begin{equation}
 E = \sum_{i = 1}^{m} \left[\frac{w(c_i)}{\upsilon(c_i)}(\mu(c_i) - 
\gamma(c_i)) + 
\gamma(c_i) T\right]
\label{label:eVarSpeed}
\end{equation}

The notation used is indicated in Table~\ref{tab:notations}.  (Though
some terms are defined later in the paper, we list them here for quick
reference.)

\begin{table}%
\begin{centering}
{%
\begin{tabular}{>{\centering}m{0.27\columnwidth}>{\centering}m{0.64\columnwidth}
}
\toprule
\textbf{Symbol} & \textbf{Meaning} \\
\midrule
$E_{m,r}$ & Total energy consumption of system with $m$ machines and jobs 
distributed to $r$ machines \\ 
$T$ & Makespan of the system \\ 
$\mathcal{C}$ & Set of machines $c_1, c_2, \ldots , c_m$ \\ 
$\mathcal{R}$ & Working set of machines \\ 
$r$ & Number of machines on which jobs are distributed \\ 
$\mu(c_i)$ & Power consumption in working state by machine $c_i$ \\ 
$\gamma(c_i)$ & Power consumption in idle state by machine $c_i$ \\ 
$\upsilon(c_i)$ & Speed of machine $c_i$ \\ 
$\tau(c_i)$ & Time spent in working state by machine $c_i$ \\ 
$\kappa(c_i)$ & Time spent in idle state by machine $c_i$ \\ 
$w(c_i)$ & Amount of work given to machine $c_i$ \\ 
$\mathcal{P}$ & Set of jobs $p_1, p_2, \ldots , p_n$ \\ 
$\psi(p_j)$ & Weight of job $p_j$ \\ 
$W$ & Total working time of all the jobs  \\ 
\bottomrule
\end{tabular}
}
\caption{Notation}
\label{tab:notations}
\end{centering}
\end{table}%

We assume the following in our system:
\begin{enumerate}
\item All jobs are independent, i.e., there are no precedence
  constraints between jobs.  It means a job need not wait for
  completion of any other job to start its execution.
  Hence, there is no idle time (gap) for a machine in between the
  execution of two jobs.  This implies that
\begin{equation}
 T = {\max}_i {\tau(c_i)}
 \label{label:tmax}
\end{equation}
\item In the class where jobs are divisible, the divisibility of jobs
  means, jobs can be arbitrarily divided and assigned to any machine.
\item All machines stay on for the duration of the makespan of the
  whole set of jobs; no machine is switched off while others continue
  to run.  At any instant while the jobs are being executed, each
  machine $c_i$ either consumes its working power $\mu(c_i)$ if
  working, or its idle power $\gamma(c_i)$ if not.  (The situation
  where a machine gets switched off to consume zero power while not
  executing a job, can of course be trivially handled by setting
  $\gamma(c_i)$ to 0.)
\item \label{assum:order}For the model in which the speed of all
  machines is the same, all machines are indexed in increasing order
  of the differences of their working power and idle power, i.e.,
  \(\mu(c_i) - \gamma(c_i) \leq \mu(c_j) - \gamma(c_j), \forall i,j\)
  where \(1\leq i \leq j \leq m\).  For the model where machines can
  have different speeds, the machines are indexed in an order such
  that the first machine is the one with smallest
  \(\frac{\mu(c_i)-\gamma(c_i)+\Gamma}{\upsilon(c_i)}\) and afterwards
  machines are indexed in non-decreasing order of
  $\frac{\mu(c_i)-\gamma(c_i)}{\upsilon(c_i)}$, where \(1\leq i \leq
  m\).
\item For the model in which the speed of all machines is the same,
  one unit of work takes one unit of time for execution, irrespective
  of the machine on which it is executed.

\end{enumerate}

The problem of energy-minimal scheduling refers to finding schedules
for a set of jobs which would require minimum total energy when
executed on a given system:

\begin{enumerate}
 \item \emph{Identical speed machines}:\\ The machines in the system
   have different working and idle power consumptions, but have equal
   speeds in execution of jobs.  The aim for this type of system is to
   minimize $E$ given by~\eqref{label:energy}.
  \item \emph{Different speed machines}:\\ Machines can have different
    power consumptions as well as different speeds.  The aim
    for this type of system is to minimize $E$ given
    by~\eqref{label:eVarSpeed}.
\end{enumerate}

For each of these classes we have two sub-classes, with divisible jobs and
non-divisible jobs.

In any of these classes, the aim is to schedule a set of jobs on the given
set of machines such that the energy consumption of the system $E$
given by~\eqref{label:energy} and~\eqref{label:eVarSpeed} is
minimized.

The minimum makespan scheduling does not guarantee minimal
energy schedules, and in fact makespan scheduling is a special
case of energy-minimal scheduling, where the working power of machines
is equal to their idle power.  Putting $\mu(c_i) = \gamma(c_i)$ in
~\eqref{label:energy}, we get,

\begin{eqnarray*}
  E &=& \sum_{i = 1}^{m} [\mu(c_i)\tau(c_i) + \gamma(c_i)(T-\tau(c_i))]\\
    &=& \sum_{i = 1}^{m} \mu(c_i)T
\end{eqnarray*}
This means energy is directly proportional to makespan, and hence
energy can be minimized by minimizing makespan if the working power is
equal to idle power.

Consider an arguably more practical setting where idle power is
less than the working power, $ \gamma(c_i)= z_i \cdot \mu(c_i)$,
where $0 \leq z_i \leq 1$:
\begin{eqnarray*}
  E &=& \sum_{i = 1}^{m} [\mu(c_i)\tau(c_i) + \gamma(c_i)(T-\tau(c_i))]\\
    &=& \sum_{i = 1}^{m} [\mu(c_i)\tau(c_i) + z_i\mu(c_i)(T-\tau(c_i))]\\
    &=& \sum_{i = 1}^{m} \mu(c_i)[\tau(c_i) + z_iT-z_i\tau(c_i)]
\end{eqnarray*}

Then, the equation of energy can be written as:
\begin{equation}
 E= \sum_{i = 1}^{m} \mu(c_i)[z_iT +\tau(c_i)(1  -z_i)]
 \label{label:powerratio}
\end{equation}

Hence, from~\eqref{label:powerratio} we can see that energy cannot be
minimized by just minimizing makespan, though the energy consumption
of an individual machine does depend upon its working time.  Hence
energy minimization is a more generic problem than makespan
minimization. 

In this section, we have introduced the system model and the
assumptions of the scheduling problem we have considered. We have
also indicated the objective function which we will be looking to
minimize.  We now proceed to solve scheduling problems in the
following section.

\section{Precedence of Machines}
\label{sec:method}

In this section we first find the precedence of machines, i.e., the
order in which they should be assigned work.  This precedence helps in
determining which machines should be alloted work and how much, to
reduce the overall energy consumption of the system.

With differing power specifications of machines, it stands to reason
that it may be advantageous to prefer some machines over others in
doing jobs.  We investigate the properties of the system to determine
the precedence of machines so that the energy consumed in their
execution is minimized.  The problem of finding relative distribution
of work among machines can be broken into two parts:

\begin{itemize}
\item Given the set of $m$ machines, which subset of machines should
  be allowed to work and which should remain idle for all times during
  the makespan?
\item What should be the distribution of work among the working
  machines?
\end{itemize}

There are no restrictions on the jobs except that the total load is
considered to be $W$.  This means that the results derived in this
section are applicable to all type of systems.  We first derive
results for the case where machines have identical speeds and then for
the case where they can have different speeds.

\subsection{Systems with Identical Speed Machines} 
\label{subsec:i_sp}

In this subsection we consider the relative work distribution among
machines with identical speeds such that energy consumption is
minimized.

As indicated previously, we assume that all machines are indexed in
the non-decreasing order of the differences between their working and
idle powers.  We claim that when the machines are indexed in such an
order, then the loading of machines should be such that the working
times of these machines are in non-increasing order.

\begin{lemma}
Given $\mu(c_k) - \gamma(c_k) \leq \mu(c_l) - \gamma(c_l)$, then for energy
optimality of the system, $\tau(c_k) \geq \tau(c_l) \ \forall k, l$, where $1 
\leq k
\leq l \leq m$.
\label{lem1}
\end{lemma}

\begin{proof}
We prove this by contradiction.  Given $\mu(c_k) - \gamma(c_k) \leq \mu(c_l) -
\gamma(c_l)$, then let us say that for minimal energy consumption, $\tau(c_k) <
\tau(c_l)$.

With a little rearrangement of~\eqref{label:energy} we get,
\begin{equation*}
E = \sum_{i = 1}^{m} [(\mu(c_i)-\gamma(c_i))\tau(c_i) + \gamma(c_i)T]
\label{label:energy1}
\end{equation*}

This in turn gives:
\begin{equation}
\begin{split}
E = &\sum_{i = 1}^{k-1} [(\mu(c_i)-\gamma(c_i))\tau(c_i) + \gamma(c_i)T]
	  +[(\mu(c_k)-\gamma(c_k))\tau(c_k) + \gamma(c_k)T] \\
	+&\sum_{i = k+1}^{l-1} [(\mu(c_i)-\gamma(c_i))\tau(c_i) + \gamma(c_i)T]
	  +  [(\mu(c_l)-\gamma(c_l))\tau(c_l) + \gamma(c_l)T] \\
	+ &\sum_{i = l+1}^{m} [(\mu(c_i)-\gamma(c_i))\tau(c_i) + \gamma(c_i)T] 
\end{split}
\label{label:ensp}
\end{equation}

The two possible cases can be: when \(\tau(c_k) \geq \tau(c_l)\), and
when \(\tau(c_k) < \tau(c_l)\). We derive the energy equations for
both cases and then compare them.  By comparing the energy equations
in both, we arrive at a condition under which the energy consumed in
one case is less, and find a contradiction.

Case 1: $\tau(c_k) \geq \tau(c_l)$.  Take $\tau(c_k) = t +\epsilon_1$ and 
$\tau(c_l) = t
-\epsilon_1$, where $\epsilon_1 \geq 0$.

Putting these values in~\eqref{label:ensp} we get,
\begin{equation}
\begin{split}
E = &\sum_{i = 1}^{k-1} [(\mu(c_i)-\gamma(c_i))\tau(c_i) + \gamma(c_i)T]
	+[(\mu(c_k)-\gamma(c_k))(t +\epsilon_1) \\
	+& \gamma(c_k)T] 
	+\sum_{i = k+1}^{l-1} [(\mu(c_i)-\gamma(c_i))\tau(c_i) + \gamma(c_i)T]
	+ \sum_{i = l+1}^{m} [(\mu(c_i)\\
	-&\gamma(c_i))\tau(c_i) + \gamma(c_i)T] 
	+ [(\mu(c_l)-\gamma(c_l))(t -\epsilon_1) + \gamma(c_l)T] 
\end{split}
\label{label:ensp1}
\end{equation}

Case 2: $\tau(c_k) < \tau(c_l)$. Take $\tau(c_k) = t -\epsilon_2$ and 
$\tau(c_l) 
= t +\epsilon_2$,
where $\epsilon_2 > 0$.

Putting these values in~\eqref{label:ensp} we get,
\begin{equation}
\begin{split}
E' = &\sum_{i = 1}^{k-1} [(\mu(c_i)-\gamma(c_i))\tau(c_i) + \gamma(c_i)T]
    + [(\mu(c_k)-\gamma(c_k))(t -\epsilon_2) \\
    +& \gamma(c_k)T] 
   + \sum_{i = k+1}^{l-1} [(\mu(c_i)-\gamma(c_i))\tau(c_i) + \gamma(c_i)T]
    + [(\mu(c_l)\\
    -&\gamma(c_l))(t +\epsilon_2) + \gamma(c_l)T] 
   + \sum_{i = l+1}^{m} [(\mu(c_i)-\gamma(c_i))\tau(c_i) + \gamma(c_i)T] 
\end{split}
\label{label:ensp2}
\end{equation}

Subtracting~\eqref{label:ensp2} from~\eqref{label:ensp1} to compare
energies in both the cases, we get,
\begin{equation}
\begin{split}
E-E' = &[(\mu(c_k)-\gamma(c_k))(t +\epsilon_1) + \gamma(c_k)T]\\
	&- [(\mu(c_k)-\gamma(c_k))(t -\epsilon_2) + \gamma(c_k)T] \\
	 &+ [(\mu(c_l)-\gamma(c_l))(t -\epsilon_1) + \gamma(c_l)T] \\
	 &- [(\mu(c_l)-\gamma(c_l))(t +\epsilon_2) + \gamma(c_l)T]
\end{split}
\label{label:enspdif}
\end{equation}

Simplifying~\eqref{label:enspdif}, we get,

\begin{equation*}
E-E' = (\epsilon_1 + \epsilon_2)[(\mu(c_k)-\gamma(c_k)) - 
(\mu(c_l)-\gamma(c_l))]
\end{equation*}

If we say that in Case 2 the energy consumed is less, this means
\begin{eqnarray*}
 E-E'>0\\
 (\epsilon_1 + \epsilon_2)[(\mu(c_k)-\gamma(c_k)) - (\mu(c_l)-\gamma(c_l))] > 
0\\
\end{eqnarray*}
As $\epsilon_1 > 0$ and $ \epsilon_2 >0$, hence,
\begin{equation}
(\mu(c_k)-\gamma(c_k)) > (\mu(c_l)-\gamma(c_l)) 
\label{label:15} 
\end{equation}

\eqref{label:15} is in contradiction to our assumption.  QED.
\end{proof} 

We show which machines should be given comparatively more work then
others.  There is a special case in which the idle power of a machine
is proportional to its working power.  The following corollary covers
machine precedence in such a case.
\begin{corollary}
When the ratio between the working power consumption ($\mu(c_i)$) and
idle power consumption ($\gamma(c_i)$) of the machine is some constant $z
, \forall i, 1 \leq i \leq m$, then $\mu(c_i) - \gamma(c_i)$ is
proportional to $\mu(c_i)$, so we need to index the machines in the
order of $\mu(c_i)$.
\end{corollary}

For energy optimality it may well be suitable to give work to only
some of the machines while letting others run completely idle.  We now
state the condition showing which machines should be used when using
only a subset of all the machines is beneficial.

\begin{lemma} \label{lem2}
If we give work to only some $r$ machines, where $1 \leq r \leq m$,
then, for reduced energy consumption, these are the $r$ machines that
form the set $\{c_1, c_2, \ldots, c_r\}$ given our index order.
\end{lemma}

\begin{proof}
We prove this by contradiction. For the $r$ machines working, $\tau(c_i)>0$, 
and for the other $m-r$
machines $\tau(c_i) = 0$.  If $\tau(c_{r+1}) > 0$ then there must be
any $\tau(c_i)$ from the set $\{\tau(c_1),\tau(c_2),\ldots,
\tau(c_r)\}$ which is equal to $0$. But by Lemma~\ref{lem1},
$\tau(c_1) \geq \tau(c_2) \geq \tau(c_3) \geq \ldots \geq \tau(c_r) >
0$ since $\tau(c_{r+1}) > 0$. These two statements are contradictory
and so it is not possible that $\tau(c_{r+1}) > 0$.  Hence it stands
proved that for energy optimality, if we give work to only some $r$
machines where $1 \leq r \leq m$, then these $r$ machines are of the
set $\{c_1, c_2, \ldots, c_r\}$. QED.
\end{proof}

Using Lemma~\ref{lem2}, given some number of machines to be used, we
can find which machines should be assigned jobs.  We now state and
prove the condition which decides how many machines should be used so
that energy consumption of the system is minimal.

Let the energy consumption of a system of $m$ machines, when jobs are 
distributed to $r$ machines, be given by $E_{m,r}$. If we add one more machine 
to the working set of machines, then we take some amount of work from 
previously 
working machine(s) and give that amount of work to the new machine. $s(c_i)$ 
represents the amount of work taken away from machine $c_i$ and given to other 
machine(s).

\begin{theorem}
When $r-1$ machines are working, for machine $c_r$ to be given work
(i.e., $\tau(c_r) \neq 0$) and result in reduced energy consumption,
the following must hold.
\begin{equation}
\sum_{i = 1}^{r-1}[(\mu(c_i) -\gamma(c_i)-\mu(c_r)+\gamma(c_r))s(c_i)] + 
s(c_1)\sum_{i = 1}^{m}
\gamma(c_i) >0
\label{label:th}
\end{equation}
\label{th4}
\end{theorem}
 \begin{proof}
We prove this by construction. 
If only $r$ of the $m$ machines are used and the rest
remain idle all the time, \eqref{label:energy} can be re-written
as follows, where the energy is expressed by $E_{m,r}$:
 \begin{equation*}
 \begin{split}
  E_{m,r} = &\sum_{i = 1}^{r-1} [\mu(c_i)\tau(c_i) + \gamma(c_i)(T-\tau(c_i))] 
\\
&+ \mu(c_r)\tau(c_r) + \gamma(c_r)(T-\tau(c_r)) +\sum_{i = r+1}^{m} \gamma(c_i)T
 \end{split}
 \end{equation*}

Putting $T=\tau(c_1)$, [by~\eqref{label:tmax}, Lemma~\ref{lem1} and
  Lemma~\ref{lem2}]
\begin{equation}
 \begin{split}
E_{m,r} = &\sum_{i = 1}^{r-1} [\mu(c_i)\tau(c_i) + 
\gamma(c_i)(\tau(c_1)-\tau(c_i))] \\
&+ \mu(c_r)\tau(c_r) +
\gamma(c_r)(\tau(c_1)-\tau(c_r))
+
\sum_{i = r+1}^{m} \gamma(c_i)\tau(c_1)  
 \end{split}
 \label{eq:ere}
\end{equation}
Re-arranging~\eqref{eq:ere}, 
\begin{equation}
 E_{m,r} = \sum_{i = 1}^{r-1} (\mu(c_i)-\gamma(c_i))\tau(c_i) + (\mu(c_r) - 
\gamma(c_r))\tau(c_r) +\tau(c_1)\sum_{i =1}^{m}\gamma(c_i)
\label{label:trv}
   \end{equation}
   
   Also, we know that,
\begin{equation}
 W = \sum_{i = 1}^{r}\tau(c_i)
 \label{eq:w}
\end{equation}

Taking out the term $r$ from~\eqref{eq:w},
\begin{equation*}
 W = \sum_{i = 1}^{r-1}\tau(c_i) + \tau(c_r) 
\end{equation*}

\begin{equation}
\tau(c_r) = W - \sum_{i = 1}^{r-1}\tau(c_i)
\label{label:trvalue}
\end{equation}

Putting the value of $\tau(c_r)$ from~\eqref{label:trvalue}
in~\eqref{label:trv}, we get,
\begin{equation}
\begin{split}
  E_{m,r} = &\sum_{i = 1}^{r-1} (\mu(c_i)-\gamma(c_i))\tau(c_i) + (\mu(c_r) \\
  &- \gamma(c_r))(W -\sum_{i=1}^{r-1}\tau(c_i)) + \tau(c_1)\sum_{i = 
1}^{m}\gamma(c_i)   
\end{split}
\label{eq:eere}
\end{equation}

Re-arranging~\eqref{eq:eere}, 
\begin{equation}
\begin{split}
 E_{m,r} = &\sum_{i = 1}^{r-1} [(\mu(c_i)-\gamma(c_i) - \mu(c_r) + 
\gamma(c_r))\tau(c_i) ] \\
 &+ (\mu(c_r)-\gamma(c_r))W  +\tau(c_1)\sum_{i = 1}^{m}\gamma(c_i)
\end{split}
 \label{label:enf}
\end{equation}
    
This is a general equation for the energy of a system with $m$
machines, with jobs given to $r$ machines, with $1 \leq r \leq m$ and
$W$ is the total working time to complete all jobs.

If we give jobs to $r-1$ machines in a $m$ machine system, its energy
can be derived by putting $r=r-1$ in~\eqref{label:enf}.  Also the
$\tau(c_i)$s are changed to $\tau(c_i)'$s as in the following:
\begin{equation}
\begin{split}
E_{m,r-1} = &\sum_{i = 1}^{r-2} [\mu(c_i) 
-\gamma(c_i)-\mu(c_{r-1})+\gamma(c_{r-1})]\tau(c_i)'
\\&+
(\mu(c_{r-1}) -
\gamma(c_{r-1}))W 
 + \tau(c_1)'\sum_{i = 1}^{m} \gamma(c_i)
\label{label:enf1}
\end{split}
\end{equation}

It is better that we give work to more machines only when the energy
consumption by doing so is lowered.  Hence, if $E_{m,r-1} > E_{m,r}$,
then only we should give work to $r$ machines, else we give to $r-1$
machines only.  Thus the condition for using $r$ machines can be
written as:
\begin{equation}
 E_{m,r-1} - E_{m,r} > 0 
 \label{label:n10}
\end{equation}

Using~\eqref{label:enf} and~\eqref{label:enf1}, 
\begin{equation}
\begin{split}
 E_{m,r-1} -E_{m,r} =& \sum_{i = 1}^{r-1} [(\mu(c_i)
-\gamma(c_i)-\mu(c_r)+\gamma(c_r))(\tau(c_i)'-\tau(c_i))]\\
& + (\tau(c_1)'-\tau(c_1))\sum_{i = 1}^{m} \gamma(c_i)
\end{split}
  \label{label:n16}
\end{equation}

Taking $\tau(c_i)'-\tau(c_i)=s(c_i)$, we get,
\begin{equation*}
 E_{m,r-1} -E_{m,r} = \sum_{i = 1}^{r-1} [(\mu(c_i) 
-\gamma(c_i)-\mu(c_r)+\gamma(c_r))s(c_i)]
 + s(c_1)\sum_{i = 1}^{m} \gamma(c_i)
  \label{label:n17}
\end{equation*}
Here, $s(c_i)$ informally signifies the amount of work taken away from
machine $c_i$, to facilitate giving work to the machine $c_r$.

To give the jobs to $r$ machines $E_{m,r-1} -E_{m,r} >0$.
\begin{equation*}
 \sum_{i = 1}^{r-1} [(\mu(c_i) -\gamma(c_i)-\mu(c_r)+\gamma(c_r))s(c_i)]
 + s(c_1)\sum_{i = 1}^{m} \gamma(c_i) > 0
 \label{label:th4e}
\end{equation*}
QED.
\end{proof}

Having derived the condition describing which machines to use for
doing jobs, we now find out the amount of
work to be given to each of these machines.  

\begin{theorem}
If we give jobs to $r > 1$ machines (where $r$ is number of working machines), then for minimum energy
consumption, the distribution of jobs on all $r$ machines should be
equal and given by $\frac{W}{r}$.\\
\label{th5}
\end{theorem}

\begin{proof}
We prove this by contradiction.  Consider two cases, one in which the
distribution of work among the machines, which qualify to work
according to Theorem~\ref{th4} is equal, and the other in which the
work distribution is unequal (only two machines have given different/unequal amount of work and that amount of work differs by $\epsilon$. As we show, even if two machines differ in amount of work given to them, in that case also the energy is increased).  We claim, to show the contradiction,
that in the case in which the distribution is unequal, the energy
consumption is greater as compared to the other, and that such is
hence a non-optimal distribution.

Case 1: Equal distribution, $\tau(c_i)= \frac{W}{r}, \forall i, 1 \leq i \leq 
r$,
i.e., $\tau(c_1) = \tau(c_2) = \tau(c_3) = \ldots = \tau(c_r) = \frac{W}{r}$.

Case 2: Unequal distribution, $\tau(c_1) = \frac{W}{r} + \epsilon$, $\tau(c_2) =
\tau(c_3) = \ldots = \tau(c_{r-1}) = \frac{W}{r}$ and $\tau(c_r)= \frac{W}{r} -
\epsilon$.

Here we assign more work to machine 1 compared to that required by
equal distribution of work among $r$ machines (by Lemma~\ref{lem1}, a
bias has to favor smaller-numbered machines, and therefore machine 1
most of all).  The amount of extra work given to machine 1 is
$\epsilon$ and this amount of work is withdrawn from machine $c_r$.
We chose only to alter the work distributions of machines 1 and $r$ to
keep the proof simple.  We could have chosen any machine $c_k$ for
doing extra work and any machine $c_l$ for doing less work, where $k
> l$.  But for allotting extra work $\epsilon$ to machine $c_k$, we
have to increase the amount of work assigned to all machines in the
set $\{c_1, c_2, \ldots , c_k\}$ by at least $\epsilon$, so that
Lemma~\ref{lem1} is satisfied.

Similarly for reducing the work of machine $c_l$ by $\epsilon$, we have
to reduce the amount of work assigned to all machines in the set
$\{c_{l+1}, c_{l+2}, \ldots, c_r\}$ by at least $\epsilon$ so that 
Lemma~\ref{lem1}
is satisfied.  Hence if we want to alter the amount of work of just
two machines (to keep the proof simple) from an equal distribution,
then we have to alter it for machines $c_1$ and $c_r$.  It may be noted that
our proof is generalizable for imbalances involving any number of
machines.

Let the energy in Case 1 be denoted by $E_{m,r}'$ and the energy in
Case 2 be $E_{m,r}''$.  Then according to our proposition:
\begin{equation*}
 E_{m,r}''<E_{m,r}'
\end{equation*}
\begin{equation}
 E_{m,r}''-E_{m,r}'<0
 \label{label:con}
\end{equation}
Rewriting~\eqref{label:enf} after expanding, we get
\begin{equation}
\begin{split}
E_{m,r} = & (\mu(c_1)-\gamma(c_1) - \mu(c_r) + \gamma(c_r))\tau(c_1) \\
&+ \sum_{i = 2}^{r-1} [(\mu(c_i)-\gamma(c_i) - \mu(c_r) + \gamma(c_r))\tau(c_i) 
]\\
&+ (\mu(c_r) - \gamma(c_r))W + \tau(c_1)\sum_{i = 1}^{m}\gamma(c_i) 
\end{split}
\label{label:enfull}
\end{equation}
Putting the respective values of $\tau(c_i)$'s in~\eqref{label:enfull} to
derive $E_{m,r}'$,
\begin{equation}
\begin{split}
E_{m,r}' = & (\mu(c_1)-\gamma(c_1) - \mu(c_r) + 
\gamma(c_r))\left(\frac{W}{r}\right) \\
&+ \sum_{i = 2}^{r-1} \left[(\mu(c_i)-\gamma(c_i) - \mu(c_r) + 
\gamma(c_r))\left(\frac{W}{r}\right) \right]\\ 
&+ (\mu(c_r) - \gamma(c_r))W + \left(\frac{W}{r}\right)\sum_{i = 
1}^{m}\gamma(c_i) 
\end{split}
 \label{label:enfull1}
\end{equation}
Similarly, we may derive $E_{m,r}''$ as:
\begin{equation}
 \begin{split}
E_{m,r}'' = & (\mu(c_1)-\gamma(c_1) - \mu(c_r) + \gamma(c_r))\left(\frac{W}{r}+ 
\epsilon\right) \\
&+ \sum_{i = 2}^{r-1} \left[(\mu(c_i)-\gamma(c_i) - \mu(c_r) + 
\gamma(c_r))\left(\frac{W}{r}\right) \right]\\ 
&+ (\mu(c_r) - \gamma(c_r))W + \left(\frac{W}{r}+ \epsilon\right)\sum_{i = 
1}^{m}\gamma(c_i) 
 \end{split}
  \label{label:enfull2}
\end{equation}
Subtracting~\eqref{label:enfull1} from~\eqref{label:enfull2}, we get
\begin{equation}
E_{m,r}''- E_{m,r}' =\epsilon[(\mu(c_1)-\gamma(c_1) - \mu(c_r) + \gamma(c_r))+ 
\sum_{i =
1}^{m}\gamma(c_i) ]
\label{label:enfulld}
\end{equation}
Using~\eqref{label:con} and~\eqref{label:enfulld}, we get
\begin{equation}
 (\mu(c_1)-\gamma(c_1) - \mu(c_r) + \gamma(c_r))+ \sum_{i = 1}^{m}\gamma(c_i)<0
 \label{label:cone}
\end{equation}
Now from Theorem~\ref{th4},
\begin{equation*}
\sum_{i = 1}^{r-1}[(\mu(c_i) -\gamma(c_i)-\mu(c_r)+\gamma(c_r))s(c_i)] + 
s(c_1)\sum_{i = 1}^{m}
\gamma(c_i) >
0
\end{equation*}
Re-writing the above equation, we get,
\begin{equation}
\begin{split}
&(\mu(c_1) -\gamma(c_1)-\mu(c_r)+\gamma(c_r))s(c_1)+\sum_{i = 2}^{r-1} 
[(\mu(c_i)\\
&-\gamma(c_i)-\mu(c_r)+\gamma(c_r))s(c_i)] + s(c_1)\sum_{i = 1}^{m} \gamma(c_i) 
> 0
\end{split}
 \label{label:th4ea}
\end{equation}

Now $s(c_i) \geq 0, \forall i \in \{1,2,\ldots,r\}$ and according to our
Assumption~\ref{assum:order} in Section~\ref{sec:pf},
$(\mu(c_i)-\gamma(c_i)-\mu(c_r)+\gamma(c_r)) \leq 0, \forall i \in
\{1,2,\ldots,r\}$. Hence the first and the second term
of~\eqref{label:th4ea} are negative, and the third term is positive.
So for the condition in~\eqref{label:th4ea} to hold, the following is
necessary:
\begin{equation*}
(\mu(c_1) -\gamma(c_1)-\mu(c_r)+\gamma(c_r))s(c_1)+ s(c_1)\sum_{i = 1}^{m} 
\gamma(c_i) > 0
 \label{label:th4eb1}
\end{equation*}

This means,
\begin{equation}
(\mu(c_1) -\gamma(c_1)-\mu(c_r)+\gamma(c_r))+\sum_{i = 1}^{m} \gamma(c_i) > 0.
 \label{label:th4eb}
\end{equation}

But~\eqref{label:cone} is in contradiction with~\eqref{label:th4eb}.

Hence, $E_{m,r}''-E_{m,r}'<0$ is false, which means, $E_{m,r}' <
E_{m,r}''$.  Hence it stands proved that we should give equal
fractions of jobs to all the machines to get minimum energy. QED.
\end{proof}

Thus, for energy-minimal scheduling in this setting, work has to be
divided equally among $r$ machines, where $r$ is chosen to
satisfy~\eqref{label:th}. Using these results we may give a more
general statement.

\begin{corollary}
\label{cor:equal_time}
When a given set of machines are chosen for work, then all those
machines should be in working state for an equal length of time for
energy-minimal scheduling.
\end{corollary}

In the next subsection, we analyse the case when the speeds of
machines are also different.  That case is more general and includes
the current case of identical speeds, but the results for the
identical speed case help develop theorems for variable speed system.

\subsection{Systems with Different Speed Machines}
\label{subsec:var_sp}

In the previous subsection we gained insight on how a subset of
machines can prove to be more advantageous for working for the
complete makespan, while the rest of the machines better be idle for
all times.  In this subsection we use the results of previous
subsection to develop a theory for systems in which machines can have
different speeds.

Like previous subsection, let us assume that work is assigned to only a subset 
$\mathcal{R}$ out of $m$ machines, where 
$|\mathcal{R}| = r$. Hence,
\begin{equation}
\tau(c_i) > 0, \forall i \in \mathcal{R}
\label{eq:nonzero_t}
\end{equation}
\begin{equation}
\tau(c_k) = 0, \forall k \notin \mathcal{R}
\label{eq:equal_t2}
\end{equation}

With a little rearrangement of~\eqref{label:energy} we get,
\begin{equation}
E = \sum_{i = 1}^{m} [(\mu(c_i)-\gamma(c_i))\tau(c_i) + \gamma(c_i)T]
\label{label:energy5_1}
\end{equation}

Using~\eqref{eq:nonzero_t} and~\eqref{eq:equal_t2} in~\eqref{label:energy5_1}, 
we have, 
\begin{equation}
 E_{m,r} = \sum_{i \in 	\mathcal{R}} [(\mu(c_i)-\gamma(c_i))]\tau(c_i) + 
\sum_{i = 1}^{m}\gamma(c_i)T
\label{eq:eqg5}
\end{equation}

Here the makespan $T$ is defined as:
\begin{equation}
T = \max(\tau(c_i)), \forall i \in \mathcal{R}
\label{eq:nonzero_ms}
\end{equation}

Putting the value of  $\sum_{i = 1}^{m}\gamma(c_i) = \Gamma$ in~\eqref{eq:eqg5},
\begin{equation}
E_{m,r} = \sum_{i \in 	\mathcal{R}} [(\mu(c_i)-\gamma(c_i))]\tau(c_i) + \Gamma 
T
\label{eq:EVarTempp}
\end{equation}

Our aim is to find out the set $\mathcal{R}$ for which $E_{m,r}$ is minimum. We 
first find out that if all the work has to be assigned to only one
machine, then which will that machine be.  We present the answer in our
first lemma of this section.

\begin{lemma}
If all the work is assigned to one machine $c_i$, i.e.,
$|\mathcal{R}| = 1$ and $\mathcal{R} = \{c_i\}$, then for energy
minimality this should be the machine with minimum
$\frac{\mu(c_i)-\gamma(c_i)+\Gamma}{\upsilon(c_i)}$.
\label{lemm1}
\end{lemma}

\begin{proof}
Since all the work is assigned to only one machine $c_i$, the makespan will be, 
\begin{equation}
T = \tau(c_i) = \frac{W}{\upsilon(c_i)}
\label{eq:ms_single}
\end{equation}
Substituting these values of $\mathcal{R}$, $T$ and $\tau(c_i)$ in 
~\eqref{eq:EVarTempp}, we get,
\begin{eqnarray}
E_{m,1} &=& \frac{W(\mu(c_i)-\gamma(c_i))}{\upsilon(c_i)} +  
\frac{W\Gamma}{\upsilon(c_i)}\\
&=& \frac{W(\mu(c_i)-\gamma(c_i)+\Gamma)}{\upsilon(c_i)}
\label{eq:lm1p}
\end{eqnarray}

Since $W$ does not depend on machines, to minimize energy consumption 
$E_{m,1}$, we need to choose
$c_i$ for which $\frac{\mu(c_i)-\gamma(c_i)+\Gamma}{\upsilon(c_i)}$ is
minimum. QED.
\end{proof}

As mentioned in Assumption~\ref{assum:order} in Section~\ref{sec:pf},
we order the machines such that the first machine is with minimum
$\frac{\mu(c_i)-\gamma(c_i)+\Gamma}{\upsilon(c_i)}$. Hence here $i =
1$, and the first machine is given work in this case.

Now before moving on to find out that which other machines should be included 
in working set, we will formulate the amount of work that should be given to 
each machine in the working set. 

\begin{lemma}
If we give jobs to a set $\mathcal{R}$ of machines, then for minimum energy
consumption, the working time of all $r$ machines should be
equal and given by $\frac{W}{\sum_{i \in \mathcal{R}} \upsilon(c_i)}$.\\
\label{lemma:same_time}
\end{lemma}

\begin{proof}
We prove this lemma by induction. If working time for all the machines in set 
$\mathcal{R}$ is equal then,
\begin{equation}
\tau(c_i) = \tau(c_j) = \frac{W}{\sum_{i \in \mathcal{R}} \upsilon(c_i)},\quad 
\forall i,j \in \mathcal{R}
\label{eq:eq_ms}
\end{equation}

From~\eqref{eq:nonzero_ms},
\begin{equation}
T =  \frac{W}{\sum_{i \in \mathcal{R}} \upsilon(c_i)}
\label{eq:var_makespan}
\end{equation} 

Using~\eqref{eq:eq_ms} and~\eqref{eq:var_makespan} in~\eqref{eq:EVarTempp},
\begin{equation}
 E_{m,r} = W \left [ \frac{\sum_{i \in \mathcal{R}} (\mu(c_i)-\gamma(c_i)) + 
\Gamma}{\sum_{i \in 
\mathcal{R}} \upsilon(c_i)}\right ]
\label{eq:EVarTemp2p}
\end{equation}

\textbf{Basis:} As the basis step of induction, we show~\eqref{eq:EVarTemp2p} 
and hence our lemma holds for $|\mathcal{R}| = 1$. Substituting $\mathcal{R} 
= \{c_i\}$ in~\eqref{eq:EVarTemp2p},
\begin{equation*}
 E_{m,1} = \frac{W(\mu(c_i)-\gamma(c_i)+\Gamma)}{\upsilon(c_i)}
\end{equation*}
This is the same as~\eqref{eq:lm1p}. Thus it has been shown that basis step of 
induction holds.

\textbf{Induction Step:} In this step we show that if ~\eqref{eq:EVarTemp2p} 
holds for $|\mathcal{R}| = r$, then it will also hold for $|\mathcal{R'}| = 
r+1$. This means if we include another machine in the working set $\mathcal{R}$ 
to minimize the energy consumption, then all the machines in the new working 
set $\mathcal{R'}$ should be working for equal amount of time. Mathematically, 
the new makespan should be given by:
\begin{equation}
T_{r+1} =  \frac{W}{\sum_{i \in \mathcal{R'}} \upsilon(c_i)}
\label{eq:ms_r1}
\end{equation}

For ease of representation, let 
\begin{equation}
p_r = \sum_{i \in \mathcal{R}}(\mu(c_i)-\gamma(c_i))+\Gamma
\label{eq:clr1}
\end{equation}
and 
\begin{equation}
q_r = \sum_{i \in \mathcal{R}} \upsilon(c_i)
\label{eq:clr2}
\end{equation}

\eqref{eq:EVarTemp2p} can be written as:
\begin{equation}
 E_{m,r} = W \left [ \frac{p_r}{q_r}\right ]
\label{eq:EVarTemp4p}
\end{equation}

Let the working time of machine $c_{r+1}$ be $\tau(c_{r+1})$. The updated 
working time of all the machines in set $\mathcal{R}$ will be:
\begin{equation}
T_{r+1} = \frac{W - \tau(c_{r+1})\upsilon(c_{r+1})}{q_r}
\label{eq:r1_ms_temp}
\end{equation}
We can assume $\tau(c_{r+1}) \leq T_{r+1}$, because if it is not the case then 
we swap machine $c_{r+1}$ with any other machine in $\mathcal{R}$ which holds 
the assumption true. And since $\tau(c_{r+1}) \leq T_{r+1}$, the new makespan 
of the system will be $T_{r+1}$.

Using these values in equation~\eqref{eq:EVarTempp},

\begin{eqnarray*}
E_{m,r+1} &=& [ \sum_{i \in \mathcal{R}} (\mu(c_i)-\gamma(c_i)) + 
\Gamma]T_{r+1} + (\mu(c_{r+1})-\gamma(c_{r+1}))\tau(c_{r+1})\\
&=& p_r T_{r+1} + (\mu(c_{r+1})-\gamma(c_{r+1}))\tau(c_{r+1})
\label{eq:eq_ms_st1}
\end{eqnarray*}

Substituting $T_{r+1}$ from~\eqref{eq:r1_ms_temp} in above equation,
\begin{eqnarray*}
E_{m,r+1} &=& p_r \left [ \frac{W - \tau(c_{r+1})\upsilon(c_{r+1})}{q_r} 
\right] + (\mu(c_{r+1})-\gamma(c_{r+1}))\tau(c_{r+1})\\
&=& \frac{W p_r}{q_r} - \frac{p_r \tau(c_{r+1})\upsilon(c_{r+1})}{q_r} + 
(\mu(c_{r+1})-\gamma(c_{r+1}))\tau(c_{r+1})\\
&=& E_{m,r} - \tau(c_{r+1}) \left [ \frac{p_r \upsilon(c_{r+1})}{q_r} - 
(\mu(c_{r+1})-\gamma(c_{r+1})) \right ]
\label{eq:simp_eq_ms_ds}
\end{eqnarray*}
Hence,
\begin{equation}
E_{m,r} - E_{m,r+1} = \tau(c_{r+1}) \left [ \frac{p_r \upsilon(c_{r+1})}{q_r} - 
(\mu(c_{r+1})-\gamma(c_{r+1})) \right ]
\label{eq:eq_ms_ds}
\end{equation}
Since the energy consumption of the system has to decrease by addition of 
machine $c_{r+1}$ in working set, i.e., $E_{m,r} - E_{m,r+1} > 0$, the R.H.S. 
of~\eqref{eq:eq_ms_ds} has to be positive. Since time cannot be negative, both 
the sub-terms $\tau(c_{r+1})$ and $\left [ \frac{p_r \upsilon(c_{r+1})}{q_r} - 
(\mu(c_{r+1})-\gamma(c_{r+1})) \right ]$ will be positive. The term $\left [ 
\frac{p_r \upsilon(c_{r+1})}{q_r} - (\mu(c_{r+1})-\gamma(c_{r+1})) \right ]$ is 
independent of working time allocation of machines. Hence to maximize $E_{m,r} 
- E_{m,r+1}$, we need to maximize $\tau(c_{r+1})$. As we know the maximum value 
of $\tau(c_{r+1}) = T_{r+1}$, using this in~\eqref{eq:r1_ms_temp},
\begin{eqnarray*}
T_{r+1} &=& \frac{W - T_{r+1}\upsilon(c_{r+1})}{q_r}\\
        &=& \frac{W - T_{r+1}\upsilon(c_{r+1})}{\sum_{i \in \mathcal{R}} 
\upsilon(c_i)}
\label{eq:wxyz1}
\end{eqnarray*}
Simplifying above equation,
\begin{eqnarray*}       
T_{r+1}\left (\sum_{i \in \mathcal{R}}\upsilon(c_i) \right ) &=& W - 
T_{r+1}\upsilon(c_{r+1})\\
T_{r+1} \left (\sum_{i \in \mathcal{R}}\upsilon(c_i) + \upsilon(c_{r+1}) \right 
) &=& W\\
T_{r+1}\left (\sum_{i \in \mathcal{R'}}\upsilon(c_i)\right ) &=& W\\
T_{r+1} &=& \frac{W}{\sum_{i \in \mathcal{R'}}\upsilon(c_i)}       
\label{eq:wxyz2}
\end{eqnarray*}
This shows~\eqref{eq:ms_r1} is true. Hence our induction step is also proved. 
This completes the proof that all the machines in the working set should be 
working for equal amount of time for minimal energy consumption.

\end{proof}

Using above lemma, the energy consumed by system when machines in set 
$\mathcal{R}$ are working is given by:
\begin{equation*}
 E_{m,r} = W \left [ \frac{\sum_{i \in \mathcal{R}} (\mu(c_i)-\gamma(c_i)) + 
\Gamma}{\sum_{i \in 
\mathcal{R}} \upsilon(c_i)}\right ]
\label{eq:EVarTemp2pp}
\end{equation*}
Since our goal is to find out relative distribution of work, without
any loss of generality, we can assume that the total work is one unit
to make the representation simpler.
\begin{equation}
 E_{m,r} = \frac{\sum_{i \in \mathcal{R}} (\mu(c_i)-\gamma(c_i)) + 
\Gamma}{\sum_{i \in 
\mathcal{R}} \upsilon(c_i)}
\label{eq:EVarTemp2}
\end{equation}

It can be noted here that irrespective of machines being of identical or 
different speeds, the machines in the working set should be working for equal 
amount of time  for minimal energy consumption. 

We already figured out which machine to give work to if only one machine can be 
assigned work. Now, if we want to give jobs to two machines then we will only 
do this if the energy consumed by two machines is less than the energy 
consumption of one machine. The same principle applies whenever we want to 
expand the working set of machines. The new expanded set must have lower energy 
consumption than the previous set. In the next theorem, we specify the 
conditions under which a set can be expanded.

\begin{theorem}
If a machine $c_j$ has to be included in the working set of machines 
$\mathcal{R}$, then the following condition must be satisfied
\begin{equation}
 \frac{\sum_{i \in \mathcal{R}} (\mu(c_i)-\gamma(c_i)) + \Gamma}{\sum_{i \in 
\mathcal{R}} \upsilon(c_i)} > \frac{\mu(c_j)-\gamma(c_j)}{\upsilon(c_j)}
\label{eq:con1}
\end{equation}
 \label{thm2}
\end{theorem}

\begin{proof}
Re-writing our condition~\eqref{eq:con1}, using~\eqref{eq:clr1} 
and~\eqref{eq:clr2}
\begin{equation*}
\frac{p_r}{q_r} > \frac{\mu(c_j)-\gamma(c_j)}{\upsilon(c_j)}
\label{eq:c1}
\end{equation*}
Multiplying the denominators both sides,
\begin{equation*}
p_r \upsilon(c_j) > \mu(c_j) q_r - \gamma(c_j) q_r
\label{eq:c2}
\end{equation*}
Adding $p_r q_r$ both sides,
\begin{eqnarray*}
p_r \upsilon(c_j)  + p_r q_r > \mu(c_j) q_r - \gamma(c_j) q_r + p_r q_r \\
p_r(\upsilon(c_j)  + q_r) > q_r(\mu(c_j) - \gamma(c_j) + p_r) \\
\frac{p_r}{q_r} > \frac{p_r + \mu(c_j)-\gamma(c_j)}{q_r + \upsilon(c_j)}
\label{eq:c3}
\end{eqnarray*}
Back substituting $p_r$ and $q_r$, we get
\begin{equation*}
 \frac{\sum_{i \in \mathcal{R}} (\mu(c_i)-\gamma(c_i)) + \Gamma}{\sum_{i \in 
\mathcal{R}} \upsilon(c_i)} >  \frac{\sum_{i \in \mathcal{R}} 
(\mu(c_i)-\gamma(c_i)) + 
\Gamma + \mu(c_j)-\gamma(c_j)}{\sum_{i \in \mathcal{R}} \upsilon(c_i) + 
\upsilon(c_j)}
\label{eq:Econd1}
\end{equation*}
Let  $\mathcal{R'}$ be the new set formed for including $c_j$ in  $\mathcal{R}$.
\begin{equation}
 \frac{\sum_{i \in \mathcal{R}} (\mu(c_i)-\gamma(c_i)) + \Gamma}{\sum_{i \in 
\mathcal{R}} \upsilon(c_i)} >  \frac{\sum_{i \in \mathcal{R'}} 
(\mu(c_i)-\gamma(c_i)) + 
\Gamma}{\sum_{i \in \mathcal{R'}} \upsilon(c_i)}
\label{eq:Econd2}
\end{equation}
The left hand side of~\eqref{eq:Econd2} gives the energy consumption
by working set $\mathcal{R}$. The right hand side
of~\eqref{eq:Econd2}, gives the energy consumption when machine $c_j$
is included in the original working set. Hence the energy consumption
of the new set is lower when the condition~\eqref{eq:con1} is
satisfied. QED.
\end{proof}

The previous theorem gives the condition in which a machine could be
included in the working set to reduce the energy consumption of the
system. But there can be many machines which satisfy this
condition. We need to find out the machine which reduces the energy
consumption by largest amount and hence make the system energy
minimal. Our next theorem specifies which machine should be given
preference for inclusion in working set.

\begin{theorem}
\label{thm:varSpeed2}
Given any two machines $c_j$ and $c_k$ which qualify to be included in working 
set 
$\mathcal{R}$ according to Theorem~\ref{thm2}, i.e.,
\begin{equation}
\frac{\mu(c_j)-\gamma(c_j)}{\upsilon(c_j)} < \frac{p_r}{q_r}
\label{con1}
\end{equation}
and 
\begin{equation}
\frac{\mu(c_k)-\gamma(c_k)}{\upsilon(c_k)}  <  \frac{p_r}{q_r} 
\end{equation}
For minimal energy consumption, machine $c_j$ shall be chosen over $c_k$ if $\upsilon(c_j) > \upsilon(c_k)$
and
\begin{equation}
\frac{\mu(c_j)-\gamma(c_j)}{\upsilon(c_j)} < 
\frac{\mu(c_k)-\gamma(c_k)}{\upsilon(c_k)} 
\label{eq:orderCond1}
\end{equation}
\label{thmm3}
\end{theorem}

\begin{proof}
Given the condition~\eqref{eq:orderCond1}, we would like to derive energy 
equations from it. Multiplying the denominators both sides 
in~\eqref{eq:orderCond1}, we get,
\begin{equation}
 (\mu(c_j)-\gamma(c_j))\upsilon(c_k) < (\mu(c_k)-\gamma(c_k))\upsilon(c_j)
\label{ab}
\end{equation}

By adding and subtracting $(\mu(c_j)-\gamma(c_j))\upsilon(c_j)$ in 
equation~\eqref{ab}, we get,
\begin{equation}
 \frac{(\mu(c_j)-\gamma(c_j)) - 
(\mu(c_k)-\gamma(c_k))}{\upsilon(c_j)-\upsilon(c_k)} < 
\frac{(\mu(c_j)-\gamma(c_j))}{\upsilon(c_j)}
 \label{con2}
\end{equation}
From~\eqref{con1} and~\eqref{con2}, we get,
\begin{equation*}
 \frac{(\mu(c_j)-\gamma(c_j)) - 
(\mu(c_k)-\gamma(c_k))}{\upsilon(c_j)-\upsilon(c_k)} < 
\frac{(\mu(c_j)-\gamma(c_j))}{\upsilon(c_j)} < \frac{p_r}{q_r}
 \label{con3}
\end{equation*}

\begin{equation*}
 \frac{(\mu(c_j)-\gamma(c_j)) - 
(\mu(c_k)-\gamma(c_k))}{\upsilon(c_j)-\upsilon(c_k)} < 
\frac{p_r}{q_r}
\end{equation*}

\begin{equation*}
 ((\mu(c_j)-\gamma(c_j)) - (\mu(c_k)-\gamma(c_k)) )q_r < 
p_r(\upsilon(c_j)-\upsilon(c_k))
\end{equation*}

\begin{equation}
 (\mu(c_j)-\gamma(c_j))q_r - (\mu(c_k)-\gamma(c_k)) q_r - 
p_r\upsilon(c_j)+p_r\upsilon(c_k) <0
\label{cd}
\end{equation}

Adding~\eqref{ab} and~\eqref{cd}, we get,
\begin{equation}
\begin{split}
&(\mu(c_j)-\gamma(c_j))q_r - (\mu(c_k)-\gamma(c_k)) q_r - 
p_r\upsilon(c_j)+p_r\upsilon(c_k)\\
& + (\mu(c_j)-\gamma(c_j))\upsilon(c_k) - (\mu(c_k)-\gamma(c_k))\upsilon(c_j) 
<0 
\end{split}
\label{de}
 \end{equation}

Now adding and subtracting $p_rq_r$ in~\eqref{de}, we get,
\begin{equation}
 \frac{p_r+\mu(c_j)-\gamma(c_j)}{q_r+\upsilon(c_j)} < 
\frac{p_r+\mu(c_k)-\gamma(c_k)}{q_r+\upsilon(c_k)}
 \label{eq:condOut1}
\end{equation}
In~\eqref{eq:condOut1}, left hand side gives the energy consumption of the 
system when machine $c_j$ is included in the working set $\mathcal{R}$, while 
the right hand side gives the energy consumption of the system when machine 
$c_k$ is included. The energy consumption is lower with machine $c_j$ in 
comparison to machine $c_k$ and hence machine $c_j$ would be given preference 
to be included in working set. (Note: After including $c_j$ in working set, the condition given in Theorem~\ref{thm2} is checked again).  QED.
\end{proof}

In Theorem ~\ref{thm:varSpeed2}, we considered a special condition of $\upsilon(c_j) > \upsilon(c_k)$. But there can be cases where $\frac{\mu(c_j)-\gamma(c_j)}{\upsilon(c_j)} < 
\frac{\mu(c_k)-\gamma(c_k)}{\upsilon(c_k)}$ but $\upsilon(c_j) \leq \upsilon(c_k)$. In a  rare case out of the aforementioned case, it is possible that choosing machine $c_k$ over $c_j$ will be more energy efficient. 
These special cases are fairly rare hence we would like to use Theorem ~\ref{thm:varSpeed2} for rest of our analysis. But we prove here that even if we choose $c_j$ before $c_k$ while forming our working set, $c_k$ will be definitely included in the working set later on.

\begin{theorem}
\label{thm:varSpeed2new}
Given any two machines $c_j$ and $c_k$ which qualify to be included in working 
set 
$\mathcal{R}$ according to Theorem~\ref{thm2}, i.e.,
\begin{equation}
\frac{\mu(c_j)-\gamma(c_j)}{\upsilon(c_j)} < \frac{\mu(c_k)-\gamma(c_k)}{\upsilon(c_k)} < \frac{p_r}{q_r}
\label{con1new}
\end{equation}
where $\upsilon(c_j) < \upsilon(c_k)$ and,
\begin{equation}
 \frac{p_r+\mu(c_j)-\gamma(c_j)}{q_r+\upsilon(c_j)} > 
\frac{p_r+\mu(c_k)-\gamma(c_k)}{q_r+\upsilon(c_k)}
 \label{eq:condOut1new}
\end{equation}
The optimal working set of machines will contain $c_k$.i.e.
\begin{equation}
 \frac{p_r+\mu(c_j)-\gamma(c_j)+\mu(c_k)-\gamma(c_k)}{q_r+\upsilon(c_j)+\upsilon(c_k)} < 
\frac{p_r+\mu(c_j)-\gamma(c_j)}{q_r+\upsilon(c_j)}
 \label{eq:condOut1new2}
\end{equation}
\end{theorem}

\begin{proof}
From ~\eqref{con1new},
\begin{equation}
\frac{\mu(c_k)-\gamma(c_k)}{\upsilon(c_k)} < \frac{p_r}{q_r}
\label{con1p}
\end{equation}
Multiplying the denominators both sides,
\begin{equation}
q_r(\mu(c_k)-\gamma(c_k)) <  p_r\upsilon(c_k)
\label{con1pp}
\end{equation}
Multiplying $\upsilon(c_k)(\mu(c_k)-\gamma(c_k))$ both sides,
\begin{equation}
(q_r + \upsilon(c_k))(\mu(c_k)-\gamma(c_k)) <  (p_r + (\mu(c_k)-\gamma(c_k)))\upsilon(c_k)
\label{con1ppp}
\end{equation}
\begin{equation}
 \frac{\mu(c_k)-\gamma(c_k)}{\upsilon(c_k)} < 
\frac{p_r+\mu(c_k)-\gamma(c_k)}{q_r+\upsilon(c_k)}
 \label{eq:condOut1newp}
\end{equation}

Combining ~\eqref{eq:condOut1new} and ~\eqref{eq:condOut1newp},
\begin{equation}
 \frac{\mu(c_k)-\gamma(c_k)}{\upsilon(c_k)} < 
\frac{p_r+\mu(c_j)-\gamma(c_j)}{q_r+\upsilon(c_j)}
 \label{eq:condOut1newpp}
\end{equation}
Multiplying the denominators both sides,
\begin{equation}
(\mu(c_k)-\gamma(c_k))(q_r+\upsilon(c_j)) < (p_r+\mu(c_j)-\gamma(c_j))\upsilon(c_k)
\end{equation}
adding $(p_r+\mu(c_j)-\gamma(c_j))(q_r+\upsilon(c_j)$ both sides and solving,
\begin{equation}
 \frac{p_r+\mu(c_j)-\gamma(c_j)+\mu(c_k)-\gamma(c_k)}{q_r+\upsilon(c_j)+\upsilon(c_k)} < 
\frac{p_r+\mu(c_j)-\gamma(c_j)}{q_r+\upsilon(c_j)}
 \label{eq:condOut1new2ppp}
\end{equation}
Since the energy consumption by including machine $c_k$ is lesser than just including machine $c_j$, machine $c_k$ will always be in the working set.
\end{proof}

For the rest of our paper we stick with Theorem~\ref{thm:varSpeed2} for finding precedence of machines as the case explained in Theorem~\ref{thm:varSpeed2new} is rare and will not affect our analysis. When we put Lemma~\ref{lemm1} and Theorem~\ref{thm:varSpeed2} together, we get 
the order in which machines will get preference to be alloted work. This order 
is same as mentioned in Assumption~\ref{assum:order} in 
Section~\ref{sec:pf}.
Combining Lemma~\ref{lemm1}, 
Theorem~\ref{thm:varSpeed2} and Assumption~\ref{assum:order}, we can state a 
Corollary.

\begin{corollary}
If the machines are indexed in an order such that the first machine is
the one with smallest
\(\frac{\mu(c_i)-\gamma(c_i)+\Gamma}{\upsilon(c_i)}\) and afterwards
in increasing order of $\frac{\mu(c_i)-\gamma(c_i)}{\upsilon(c_i)}$,
and if the working set contains $r$ machines, then for minimal energy
consumption $\mathcal{R} = \{c_1, c_2, \ldots, c_r\}$, where each
machine in this set is working for an equal amount of time.
\end{corollary}
Using this Corollary, \eqref{eq:EVarTemp2} can be re-written as, 
\begin{equation}
 E_{m,r} = W \left[\frac{\sum_{i = 1}^r (\mu(c_i)-\gamma(c_i)) + 
\Gamma}{\sum_{i = 
1}^r \upsilon(c_i)}\right]
\label{eq:EVarTemp5}
\end{equation}
which gives the energy consumption of a system with total work $W$. 

From~\eqref{eq:EVarTemp5} we see, energy is dependent on 
work $W$, but from~\eqref{label:th4eb} and Theorems~\ref{thm2} 
and~\ref{thm:varSpeed2} we know that finding the number of working machines 
(value of $r$) from the given set $\mathcal{C}$ of machines is independent of 
work $W$. Hence the scheduling decisions do not depend upon the quantum of the 
work given to the system.

In this section, we found the precedence amongst machines for allotting work 
based on their energy and speed specifications. The results derived in this 
section are applicable for handling any class of jobs, which means that the 
precedence of machines will remain same irrespective of whether the jobs 
are divisible, non-divisible or whether they have precedence constraints. In 
the next sections we consider different classes of systems based on type of 
jobs.

\section{Algorithms and Complexity}
\label{sec:anc}
Depending upon the jobs to be executed the systems can be classified
into many categories. We have already mentioned those categories in
Section~\ref{sec:pf}.  In the current section we analyse the
complexity of scheduling problem for these various types of systems
and give scheduling algorithms for the same.  We first analyse
systems with identical speed machines and then move ahead to systems
with different speeds.

\subsection{Systems with Identical Speed Machines}
In Section~\ref{subsec:i_sp}, we gave results which govern the
scheduling of systems with identical speed machines. Given the total
amount of work and energy specifications of machines, we can find the
number of machines which have to be assigned work using results of
Section~\ref{subsec:i_sp}.  Algorithm~\ref{algo1} is an implementation
of those results.

Algorithm~\ref{algo1} takes the number of machines, the working power
and idle power of the machines, and the total work to be done as
input, and gives the value of $r$, i.e., the number of working
machines, as output.  The machines are indexed in the precedence
order given by Assumption~\ref{assum:order} in Section~\ref{sec:pf}.
Now, we apply binary search to find the value of $r$, such that we get
minimal energy value of the system.  The algorithm is based on the
previously given results. The computation complexity of
Algorithm~\ref{algo1} is $\mathcal{O}(m)$. We now find the complexity
of scheduling problems and give algorithms for divisible and
non-divisible jobs.

 \begin{algorithm}[htbp]
\footnotesize
\LinesNumbered
\SetKwInOut{Input}{input}\SetKwInOut{Output}{output}
\Input{Number of machines ($m$), working power of machines ($\mu(c_i)$), idle 
power 
of machines ($\gamma(c_i)$), total work to be 
done ($W$)}
\Output{Number of working machines ($r$), makespan ($T$)}
\BlankLine
\For{$i = 1$ to $m$}{
calculate $\mu(c_i) - \gamma(c_i)$
}
sort $(\mu(c_i) - \gamma(c_i))$\;
$low \leftarrow 1$, $high \leftarrow m$\;
\While{$low \leq high$}{
$mid \leftarrow low + \left \lfloor \frac{(high - low)}{2} \right \rfloor$ \;
// $E(k) \leftarrow (\frac{W}{k})[\sum_{i = 1}^{k}\mu(c_i) + \sum_{i = 
k+1}^{m}\gamma(c_i)]$\;
Calculate $E(mid-1), E(mid), E(mid +1)$\;
\eIf{$(E(mid-1) > E(mid)) \&\& (E(mid) \leq E(mid +1))$}
{$r \leftarrow mid$\;
return $r$\;}
{\eIf{$E(mid-1) > E(mid)$}
{$low \leftarrow  mid$\;}{$high \leftarrow mid$\;}}}
 $T \leftarrow \frac{W}{r}$\;
\caption{Find the number of working machines of the system in order to
  minimize the energy consumption when the speeds of machines are
  identical}
\label{algo1}
\end{algorithm}

\subsubsection{Divisible Jobs}
\label{subsub:id_div}
In this case we are given a set $\mathcal{P}$ of jobs, where jobs can
be arbitrarily broken into any number of smaller jobs and these
smaller jobs can be executed in parallel on different machines. 
We take this assumption to simplify the analysis.
We consider the energy consumption overhead due to division/breaking of
jobs to be null. Since the jobs are divisible, they can be
trivially distributed over $r$ machines with makespan $T$, where $r$
and $T$ are given by Algorithm~\ref{algo1}.

Hence, the energy of system with divisible jobs can be calculated using 
following equation.

\begin{equation}
 E_{m,r} = T[\sum_{i = 1}^{r}\mu(c_i) + \sum_{i = r+1}^{m}\gamma(c_i)]
 \label{eq:eqe}
\end{equation}

Using Algorithm~\ref{algo1} and~\eqref{eq:eqe} we can find the energy
of the system.

\begin{remark}
Energy-minimal scheduling of divisible jobs on identical speed
machines can be done in linear time.
\end{remark}

We now analyse the complexity of scheduling of non-divisible jobs.

\subsubsection{Non-divisible Jobs}
\label{subsub:id_nd}
Given the set $\mathcal{P}$ of non-divisible jobs and time required
to execute job $p_j$ on a machine with unit speed $\psi(p_j)$, our aim
is to distribute the set $\mathcal{P}$ of jobs among the given set
$\mathcal{C}$ of machines such that energy consumption is
minimum. From Lemma~\ref{lemma:same_time} it is evident that
whichever machines are chosen to work, they should be working for an
equal amount of time to achieve energy minimality.  But it is not
straightforward or sometimes even possible to distribute work in such
a way when jobs are non-divisible.  We prove that it is an NP-hard
problem to calculate the minimal-energy schedule for the current case.

\begin{proposition} 
\label{thmNPIdentical} 
In a system where machines have identical speeds and jobs are
non-divisible, computing the energy-minimal schedule is an NP-hard
problem.
\end{proposition}
\eat{
\begin{proof}
Let $r$ denote the number of machines which are given work and $T$
is the makespan according to Algorithm~\ref{algo1}. Since all the
machines have identical speed, the scheduling problem is to make
exclusive subsets $P_i$ from set $\mathcal{P}$ such that the sum of
work in set $P_i$ is $T$, i.e.,
\begin{equation}
\sum_{j \in P_i} \psi(p_{j}) = T  \quad \forall i; 1\leq i \leq r
\label{eq:subset1}
\end{equation}
Given a finite set of positive numbers and another positive number
called the goal, to find the subset whose sum is closest to the goal
is well known as the classical \emph{subset-sum
  problem}~\cite{Johnson1973}.  In the current scenario, the set of
positive numbers $\{\psi(p_j): j \in \mathcal{P} \}$ has to be
distributed in subsets $P_i$ such that sum of each subset is the
positive number $T$. For each subset ($P_i$ ; $i \leq r$), this
problem of finding the exclusive subsets $P_i$ can be seen as the
classical \emph{subset-sum problem}. Since the subset-sum problem is
NP-hard~\cite{Johnson1973}, and our problem is reduced to subset-sum
problem; energy-minimal scheduling of jobs in a system where jobs are
non-divisible is NP-hard. QED.
\end{proof}
}

Since finding the energy-minimal schedule is NP-hard, we develop an
approximation algorithm to find energy efficient schedules. Though we
strive to achieve the makespan given by Algorithm~\ref{algo1}, in a
system with non-divisible jobs it is not always possible to achieve
it.  We state the theorem which specifies the ideal makespan $T_o$ and
working set of machines $\mathcal{R}_o$ for energy minimality of
system with non-divisible jobs.

\begin{lemma}
Given a system with identical speed machines and a set $\mathcal{P}$
of jobs with longest job of length $\psi(p_{max})$, the makespan is
given by $T_o = \max(T,\psi(p_{max}))$ and the working set of machines
$\mathcal{R}_o = \{c_1, c_2, \ldots, c_{r_o}\}$, where
\begin{equation}
r_o = \left \lceil \frac{W}{T_o}\right \rceil
\label{eq:newRo}
\end{equation}
\end{lemma}   

\begin{proof}
It is not possible to have makespan lower than $\psi(p_{max})$ since
jobs cannot be divided. Also, if we take $T_o$ less than $T$, then
we cannot have an energy-minimal schedule. Hence for energy minimality,
\begin{equation}
T_o = \max(T,\psi(p_{max}))
\label{eq:c4}
\end{equation}
From Corollary~\ref{cor:equal_time}, it is evident that work should be
distributed equally if possible to the working machines. Hence with a
possibly increased makespan, the number of working machines might get
decreased and given by~\eqref{eq:newRo}. Also from Lemma~\ref{lem2},
it is evident that the working set of machines is given by:
\begin{equation}
\mathcal{R}_o = \{c_1, c_2, \ldots, c_{r_o}\}
\label{eq:c5}
\end{equation}  
QED.
\end{proof}

\begin{algorithm}[htbp]
\footnotesize
\LinesNumbered
\SetKwInOut{Input}{input}\SetKwInOut{Output}{output}
\Input{Number of machines ($m$), working power of machines ($\mu(c_i)$), idle 
power 
of machines ($\gamma(c_i)$), total work to be 
done ($W$), $\mathcal{P}$, $\psi(p_j)$, $r_o$, $T_o$}
\Output{Makespan ($T$), energy ($E$)}
\BlankLine
$S \leftarrow \{ \}$ \;
\For{$i = 1$ to $m$}{
$d_i  \leftarrow \mu(c_i) - \gamma(c_i)$ \;
$S \leftarrow S \cup \{d_i\} $ \;
}
sort $S$\;
$J \leftarrow \{ \}$ \;
\For{$j = 1$ to $n$}{
$J \leftarrow J \cup \{\psi(p_j)\} $ \;
}
reverse-sort $(J)$\;
$r \leftarrow r_o$\;
$P_j \leftarrow \{ \}$\;
\For{$j=1$ to $r$}{
$P_j \leftarrow P_j \cup \{p_j\} $\;
$ b_j \leftarrow \sum_{j \in P_j} \psi(p_j) $\;
}
\For{$l=r+1$ to $n$}{
calculate $\min(b_j)$\;
$ b_u \leftarrow \min(b_j)$\;
$ u \leftarrow j $ \;

$P_j \leftarrow P_j \cup \{p_l\} $\;
$b_j \leftarrow \sum_{j \in P_j} \psi(p_j) $\;
}
\For{$i=1$ to $r$}{
assign $p_i $ to $c_i$ \quad $(\forall p_i \in P_i)$ \;
}
\For{$i=1$ to $r$}{
$\tau(c_i) \leftarrow \sum_{j \in P_j} \psi(p_j) $\;
}
sort $\tau(c_i)$\;
$T \leftarrow \max(\tau(c_i))$\;
\For{$i=1$ to $r$}{
$\kappa(c_i) \leftarrow  T-\tau(c_i)$\;
}
$T \leftarrow \max \tau(c_i)$\quad $(\forall p_i \in P_i)$ \;
$E \leftarrow \sum_{i = 1}^{r}\mu(c_i)\tau(c_i) + \sum_{i = 1}^{r}\gamma(c_i)\kappa(c_i) + T\sum_{i = 
r+1}^{m}\gamma(c_i)$\;
\caption{Approximation algorithm for energy-efficient scheduling of
  the system when the speeds of machines are identical}
\label{algo3}
\end{algorithm}

It is not always possible to keep each machine working till time
$T_o$. Some machines might work for more than $T_o$ and some might get
less than that because of non-divisibility of jobs. From
Lemma~\ref{lem1}, we know that the amount of work should be in decreasing order
by the machine's indices. Taking all this information into consideration, we 
have devised 
the approximation algorithm~\ref{algo3} which is an adaptation of list-scheduling in order of non-decreasing processing times for energy-minimal scheduling 
of non-divisible jobs.

In Algorithm~\ref{algo3}, we index the machines in non-decreasing
order of the difference of their working power and idle power. Jobs
are indexed in non-increasing order of their weight. Initially we
create $r$ empty buckets for jobs to be filled in. First $r$ jobs are
chosen and given to each $r$ bucket, such that each bucket has one
element. Now, to distribute remaining $n-r$ jobs to the buckets we
follow a sequential process.  The weighted sum of jobs of each bucket is
calculated and the next job
is included in the bucket in which the value is least. The same process is 
repeated until
all jobs are included in one or other bucket or sets. Now we index the
buckets in non-increasing order of their total weights. We assign the
jobs of set $i$ to machine $c_i$. Now the working time of all the
machines is calculated. Hence we can calculate the makespan, which is
the maximum working time of all the machines.

This approximation algorithm is inspired by Graham's 
algorithm~(\cite{graham1969})
which arranges a set of independent
non-divisible jobs on identical speed machines such that the makespan
is minimum. We made some changes in the algorithm to mould it for
different power specification of machines. We now calculate the bound
on deviation from ideal energy consumption due to approximation.

\begin{theorem}
The maximum possible ratio of energy consumption using Algorithm~\ref{algo3} 
and ideal energy consumption is given by:
\begin{equation}
\frac{{E^*}_{max}}{E_o} = 1 + 
\frac{({\frac{4}{3}}-{\frac{1}{3r_o}}-1)\Gamma}{\sum_{i = 1}^{r_o} 
(\mu(c_i)-\gamma(c_i)) + \Gamma}
\label{eq:b1}
\end{equation}
\end{theorem}
\begin{proof}
From~\eqref{label:energy}, if all the $r_o$ machines are working for
an equal amount of time $T_o$, then energy consumption is given by:
\begin{equation}
E_o = T_o[\sum_{i = 1}^{r_o} (\mu(c_i) - \gamma(c_i)) + \Gamma]
\label{eq:idealEn}
\end{equation}
From Algorithm~\ref{algo3}, the energy consumed is given by:
\begin{equation}
E^* = \sum_{i = 1}^{r_o} (\mu(c_i) - \gamma(c_i)){\tau^*}(c_i) + \Gamma T^*
\label{eq:algoOut}
\end{equation}
Assume that the algorithm gave such ${\tau^*}_i$ that all the machines 
other than $c_j$ and $c_k$ are allotted work equal to $T_o$. Mathematically,
\begin{equation}
{\tau^*}(c_j) = T_o + \epsilon,\quad {\tau^*}(c_k) = T_o - \epsilon
\label{eq:jk}
\end{equation}
where $j < k \leq r_o$. And
\begin{equation}
{\tau^*}(c_i) = T_o,\quad \forall i \neq j,k,\quad i \leq r_o 
\label{eq:inotjk}
\end{equation}
Since $j < k$, according to our algorithm ${\tau^*}(c_j) > {\tau^*}(c_k)$ and 
so $\epsilon \geq 0$. 
Using these values from~\eqref{eq:jk} and \eqref{eq:inotjk} 
in~\eqref{eq:algoOut}, 
we get,
\begin{equation}
 \begin{split}
 E^* = &(\mu(c_j) - \gamma(c_j))(T_o + \epsilon) + (\mu(c_k) - \gamma(c_k))(T_o 
- 
\epsilon)\\ + &\sum_{i \neq j,k} (\mu(c_i) - \gamma(c_i))T_o + \Gamma T^*
 \end{split}
\label{eq:algoOut1}
\end{equation}
Simplifying above,
\begin{equation}
E^* = (\mu(c_j) - \gamma(c_j)- (\mu(c_k) - \gamma(c_k)))\epsilon + \sum_{i = 
1}^{r_o} (\mu(c_i) - \gamma(c_i))T_o + \Gamma T^*
\label{eq:algoOut2}
\end{equation}
Now since $j < k$, we have $\mu(c_j) - \gamma(c_j)- (\mu(c_k) - \gamma(c_k)) 
\leq 0$. To compute the bound we are looking for the case where $E^*$ is 
maximum. That will occur when
\begin{equation}
\mu(c_j) - \gamma(c_j)- (\mu(c_k) - \gamma(c_k)) = 0
\label{eq:cond0}
\end{equation}
Using~\eqref{eq:cond0} in~\eqref{eq:algoOut2}, we get,
\begin{equation}
{E^*}_{max} = \sum_{i = 1}^{r_o} (\mu(c_i) - \gamma(c_i))T_o + \Gamma T^*
\label{eq:algoMax}
\end{equation}
Now according to~\cite{graham1969}, the bound on makespan is given 
by:
\begin{equation}
\frac{T^*}{T_o} = \frac{4}{3}-\frac{1}{3r_o}
\label{eq:LPTbound}
\end{equation} 
Using this bound in~\eqref{eq:algoMax}, we get
\begin{equation}
{E^*}_{max} = \sum_{i = 1}^{r_o} (\mu(c_i) - \gamma(c_i))T_o + \Gamma T_o 
\left(\frac{4}{3}-\frac{1}{3r_o}\right)
\label{eq:algoMax1}
\end{equation}
Simplifying above
\begin{equation}
{E^*}_{max} = T_o\left(\sum_{i = 1}^{r_o} (\mu(c_i) - \gamma(c_i)) + \Gamma  
\left(\frac{4}{3}-\frac{1}{3r_o}\right)\right)
\label{eq:algoMax11}
\end{equation}
\begin{equation}
{E^*}_{max} = T_o\left(\sum_{i = 1}^{r_o} (\mu(c_i) - 
\gamma(c_i)) + \Gamma\right) + T_o 
\Gamma\left(\frac{4}{3}-\frac{1}{3r_o} - 1\right)
\label{eq:algoMax12}
\end{equation}
Substituting the values, we get,
\begin{equation}
\frac{{E^*}_{max}}{E_o} = \frac{T_o(\sum_{i = 1}^{r_o} (\mu(c_i) - \gamma(c_i)) 
+ \Gamma]) + T_o \Gamma(\frac{4}{3}-\frac{1}{3r_o} - 1)}{T_o[\sum_{i = 1}^{r_o} 
(\mu(c_i) - \gamma(c_i)) + \Gamma]}
\label{eq:ib1}
\end{equation}
\begin{equation}
\frac{{E^*}_{max}}{E_o} = 1 + \frac{\Gamma(\frac{4}{3}-\frac{1}{3r_o} - 
1)}{\sum_{i = 1}^{r_o} (\mu(c_i) - \gamma(c_i)) + \Gamma}
\label{eq:bound1}
\end{equation}
QED.
\end{proof}

In the worst case, the values of working power and idle power are
equal, i.e.,  $\mu(c_i) = \gamma(c_i); \forall i $, so the bound is
given by,
\begin{eqnarray*}
\frac{{E^*}_{max}}{E_o}& = &1 + \frac{\Gamma(\frac{4}{3}-\frac{1}{3r_o} - 
1)}{\sum_{i = 1}^{r_o} (\mu(c_i) - \gamma(c_i)) + \Gamma}\\
 & = &1 + \frac{\Gamma(\frac{4}{3}-\frac{1}{3r_o} - 
1)}{\Gamma} \\
 & = &1 + (\frac{4}{3}-\frac{1}{3r_o} - 1) \\
  & = &\frac{4}{3}-\frac{1}{3r_o} 
\end{eqnarray*}
If $r_o$ is very large, i.e., when $r_o \to \infty$, then,
\begin{equation}
\frac{{E^*}_{max}}{E_o} = \frac{4}{3}
\label{eq:bound1num}
\end{equation}

\eqref{eq:bound1} gives the lower bound of inefficiency of our 
algorithm and~\eqref{eq:bound1num} gives the upper bound.

We now move ahead to analyse the class of scheduling problems in which
machines can have different working speeds.

\subsection{Systems with Different Speed Machines}

In this section, along with power specifications of machines we also
consider the speeds of machines to be different and give algorithms
for both divisible jobs and non-divisible jobs. In
Section~\ref{subsec:var_sp}, we proved various lemmas and theorems
which govern the scheduling of systems with different speed
machines. Using these results Algorithm~\ref{algo2} finds the set of
machines which should be assigned work for the energy-minimal
scheduling when the machines of the system have different speeds.

Algorithm~\ref{algo2} takes the number of machines, the working power,
idle power and speeds of the machines, and the total work to be done as
input, and gives the value of $r$, i.e., the number of working
machines, as output. In the algorithm, we first index the machines in
their precedence order, as given in Assumption~\ref{assum:order} in
Section~\ref{sec:pf}. Then to calculate the value of $r$, we check the
condition in~\eqref{eq:con1} from Theorem~\ref{thm2}. Then from the
given values of total work to be done and the computed value of number
of machines, we calculate the makespan. The algorithm also computes
the amount of work that has to be given to respective machines. The
complexity of Algorithm~\ref{algo2} is $\mathcal{O}(m)$.

\begin{algorithm}[htbp]
\footnotesize
\LinesNumbered
\SetKwInOut{Input}{input}\SetKwInOut{Output}{output}
\Input{Number of machines ($m$), working power of machines ($\mu(c_i)$), idle 
power of machines ($\gamma(c_i)$), sum of idle 
power of all the machines ($\Gamma$), speed of machines ($\upsilon(c_i)$), 
total work to be done ($W$)}
\Output{Number of working machines ($r$), makespan ($T$)}
\BlankLine
\For{$i = 1$ to $m$}{
calculate $\frac{\mu(c_i) - \gamma(c_i) + \Gamma}{\upsilon(c_i)}$
}
$\frac{\mu(c_1) - \gamma(c_1) + \Gamma}{\upsilon(c_1)} \leftarrow 
min(\frac{\mu(c_i) - \gamma(c_i) 
+ \Gamma}{\upsilon(c_i)})$ \;
\For{$i = 2$ to $m$}{
calculate $\frac{\mu(c_i) - \gamma(c_i)}{\upsilon(c_i)}$
}
sort $(\mu(c_i) - \gamma(c_i))$\;
\For{$r = 1$ to $m$}{
$E(r) \leftarrow [\frac{\sum_{i = 1}^{r}(\mu(c_i)-\gamma(c_i)) + 
\Gamma}{\sum_{i 
=1}^{r}\upsilon(c_i)}]W$\;
\eIf{$\frac{\mu(c_{r+1})-\gamma(c_{r+1})}{\upsilon(c_{r+1})}<E(r)$}
{$r \leftarrow r+1$\;
return $r$\;
}{stop\;}
}
$T \leftarrow \frac{W}{r}$\;
$w(c_i) \leftarrow T*\upsilon(c_i)$\;
\caption{Find the working machines and makespan of the system when 
the speeds of machines are different}
\label{algo2}
\end{algorithm}

We now find the complexity of scheduling problems and 
give algorithms for divisible and non-divisible jobs.
 
\subsubsection{Divisible Jobs}
\label{subsub:diff_div}
In this case, jobs can be arbitrarily broken into any number of
smaller jobs and these smaller jobs can be executed in parallel on
different machines.  Also we consider the energy consumption overhead
due to division/breaking of jobs overhead to be nil. In this scenario
the jobs can be trivially distributed over $r$ machines such that
machine $c_i$ gets $w(c_i)$ amount of work (as specified by
Algorithm~\ref{algo2}). The energy consumption of system with
divisible job can be given by:
\begin{equation}
E = \sum_{i = 1}^{m}\left [\mu(c_i)\frac{w(c_i)}{\upsilon(c_i)} + 
\gamma(c_i)(T-\frac{w(c_i)}{\upsilon(c_i)})\right ]
\label{eq:en_diff_div}
\end{equation}

Clearly the scheduling of divisible jobs can be done by using
Algorithm~\ref{algo2} which is a linear time algorithm.
\begin{remark}
Energy-minimal scheduling of divisible jobs on a system of machines
running at different speeds can be done in linear time.
\end{remark}

Many might argue that many systems do not have the luxury to divide
jobs arbitrarily, but it must be noted that efficiency in a system
with divisible jobs can be seen as upper bound to the achievable
efficiency. We proceed to analyse the case in which jobs are
non-divisible in the next section.

\subsubsection{Non-divisible Jobs}
\label{subsub:diff_nd}
In this section we first prove that scheduling non-divisible jobs for
energy optimality on a system with machines having different speeds is
an NP-hard problem.  We then give an approximation algorithm for the
same.

\begin{proposition} 
In a system where machines have different speeds and jobs are
non-divisible, computing energy-minimal schedule in that case is an
NP-hard problem.
\end{proposition}
\begin{proof}
The amount of work $w(c_i)$ that has to be assigned to machine $c_i$ is
given by Algorithm~\ref{algo2}. Now we need to make exclusive subsets
$P_i$ from set $\mathcal{P}$ such that the sum of work in set $P_i$ is
$w(c_i)$, i.e.,
\begin{equation}
\sum_{j \in P_i} \psi(p_{j}) = w(c_i)  \quad \forall i; 1\leq i \leq r 
\label{eq:subset11}
\end{equation}
As also discussed with Theorem~\ref{thmNPIdentical}, this is in the
form of the subset-sum problem which is NP-hard~\cite{Johnson1973}. QED.
\end{proof}

Since finding the energy-minimal schedule is NP-hard, we develop an
approximation algorithm to find energy efficient schedules. It is possible that 
due to size of the jobs, the ideal
makespan specified by Algorithm~\ref{algo2} cannot be achieved. In the next 
lemma we specify the best achievable makespan.

\begin{lemma}
Given a system of machines with different speed and a set $\mathcal{P}$
of jobs with longest job of length $\psi(p_{max})$ and the speed of the
fastest machine is given by $\upsilon_{max}$, the best possible
achievable makespan is given by:
\begin{equation}
T_o = \max(T,\frac{\psi_{max}}{\upsilon_{max}})
\label{eq:idealT0var}
\end{equation}
where $T$ is given by Algorithm~\ref{algo2}. 
\end{lemma}
\begin{proof}
If there is a job which cannot be completed within $T$, then irrespective
of the arrangement of jobs, we increase the best achievable
makespan $T_o$ to accommodate that biggest job. The biggest job cannot
be completed any earlier than
$\frac{\psi_{max}}{\upsilon_{max}}$. Also, any schedule with
makespan lower than $T$ should be suboptimal. Hence,
\begin{equation*}
T_o = \max(T,\frac{\psi_{max}}{\upsilon_{max}})
\label{eq:tt3}
\end{equation*}
QED.
\end{proof}

\begin{algorithm}[htb]
\footnotesize
\LinesNumbered
\SetKwInOut{Input}{input}\SetKwInOut{Output}{output}
\Input{Number of machines ($m$), working power of machines ($\mu(c_i)$), idle 
power of machines ($\gamma(c_i)$), speed of machines ($\upsilon(c_i)$), total 
work to be done ($W$), $\mathcal{P}$, $\psi(p_j)$, $r_o$, $T_o$}
\Output{Makespan ($T$), energy ($E$)}
\BlankLine
\For{$i = 1$ to $m$}{
$A(i) \leftarrow  \frac{\mu(c_i) - \gamma(c_i)}{\upsilon(c_i)}$ \;
}
sort $A(i)$ in non-decreasing order\;
\For{$j = 1$ to $n$}{
$B(j) \leftarrow \psi(p_j)$ \;
}
sort $B(j)$ in non-increasing order \;
$temp_i \leftarrow 0$ \;
\For{$j=1$ to $n$}{
\For{$i=1$ to $r_o$}{
$temp_i \leftarrow \frac{\psi(p_j)}{\upsilon_i} + temp_i $ \;
}
assign $p_j$ to $c_i$, where $c_i$ has $ \min(temp_i)$ \;
}
$T \leftarrow \max \tau(c_i)$\;
$E \leftarrow \sum_{i = 1}^{m}[\mu(c_i)\frac{w(c_i)}{\upsilon(c_i)} + 
\gamma(c_i)(T-\frac{w(c_i)}{\upsilon(c_i)})]$\;
\caption{Approximation algorithm for energy-efficient scheduling of
  the system when the speeds of machines are different}
\label{algo4}
\end{algorithm}

If the makespan is increased then there can be possible reduction in
number of working machines. The ideal number of working machines $r_o$
is given by the minimum number of machines which satisfy,
\begin{equation}
\frac{W}{\sum_{i=1}^{r_o}\upsilon(c_i)} \geq T_o
\label{eq:idealR0var}
\end{equation}

Using this $r_o$ as input, we present Algorithm~\ref{algo4} which
tries to distribute jobs amongst machines with different speed such that
the machines would work for an equal amount of time.

In Algorithm~\ref{algo4}, we first index the machines in non-decreasing
order of the ratio of difference of working power and idle power, and speed of 
the machines. We also index the jobs in
non-increasing order of their weights. Now, we take a job at a time
and find that on which machine it will finish first, based on the
speed of the machine. Then we assign the job to that particular
machine.  Similarly the process is repeated till all the jobs are
assigned to some machines. Now, the working time of all machines is
calculated, and the maximum working time of all machines is the
makespan.

The approximation algorithm may not achieve the exact $T_o$, but it
is within certain bounds.

The bound for the makespan given by~\cite{gonzalez1977} of such LPT
(largest processing time) schedule is
\begin{equation}
 \frac{2m}{m+1} 
 \label{eq:gb}
\end{equation}

Given this bound on makespan, we derive the bound on energy consumption
of our algorithm.

\begin{theorem}
When the speeds as well as the power specifications of the machines
are different, and when jobs are non-divisible, then the maximum possible ratio 
of energy consumption using Algorithm~\ref{algo4} and
ideal energy consumption is given by:
\begin{equation}
\frac{{E^*}_{max}}{E_o} \leq \frac{2r_o}{r_o + 1}
\label{eq:tt8}
\end{equation}
\end{theorem}

\begin{proof}
The energy consumption for ideal schedule is given by:
\begin{equation}
E_o = T_o(\sum_{i = 1}^{r_o}\mu(c_i) + \Gamma)
\label{eq:tt9}
\end{equation} 
Since we do not know which machines are working for time more than
$T_o$ and which ones less than $T_o$, for the upper bound we assume
that all machines are working for ${T^*}_{max}$. Hence,
\begin{equation}
{E^*}_{max} = {T^*}_{max}(\sum_{i = 1}^{r_o}\mu(c_i) + \Gamma)
\label{eq:tt10}
\end{equation}
Clearly,
\begin{equation}
\frac{{E^*}_{max}}{E_o} = \frac{{T^*}_{max}(\sum_{i = 1}^{r_o}\mu(c_i) + 
\Gamma)}{T_o(\sum_{i = 1}^{r_o}\mu(c_i) + \Gamma)}
\label{eq:tt11}
\end{equation}
\begin{equation}
\frac{{E^*}_{max}}{E_o} = \frac{{T^*}_{max}}{T_o}
\label{eq:tt12}
\end{equation}

By~\eqref{eq:gb},
\begin{equation}
\frac{{T^*}_{max}}{T_o} \leq \frac{2r_o}{r_o + 1}
\label{eq:tt13}
\end{equation}
From~\eqref{eq:tt12} and~\eqref{eq:tt13}, we get,
\begin{equation}
\frac{{E^*}_{max}}{E_o} \leq \frac{2r_o}{r_o + 1}
\label{eq:tt14}
\end{equation}
QED.
\end{proof}

This bound increases with number of machines in system and reaches $2$ 
asymptotically. 
Annam{\'a}ria Kov{\'a}cs~\cite{{kovacs2010new}} provided a tighter bound for LPT
makespan that is $1+\sqrt{3}/3 \approx 1.5773$.
\eat{Dobson~\cite{Dobson84} provided a tighter bound for LPT
makespan that is $\frac{19}{12}$. }
As the bound on energy depends on makespan in our algorithm, hence the worst 
case bound for our algorithm
is also $1+\sqrt{3}/3 \approx 1.5773$.

\section{Energy Efficiency: Incompatible Measures}
\label{sec:im}
There can be two measures for assessing the energy efficiency of a
system of machines: one, to consider the total energy consumed by
the system to complete some work; and the second, to consider the
fraction of the energy consumed by machines in the system to do work
over the total energy consumed in the system.

The first measure gives a sense of the energy required by the system
per unit work produced, and the second measure indicates how much of
the energy consumed by the system goes into work, and how much is idle
(non-working) consumption by the system.  We hold that the
first measure, energy per unit work, is the more meaningful one (as it
can lead to lowered overall energy consumption while completing some
amount of work). 

Our analysis clearly indicates that these two measures of energy
efficiency are incompatible, in the sense that they cannot be
simultaneously optimized for arbitrary systems.  To see why, we may
consider~\eqref{eq:EVarTemp5} for the total energy of the system.
In~\eqref{eq:EVarTemp5}, if we put $ \gamma(c_i) = 0, \forall i, 1
\leq i \leq m$, then the working energy is given by:
\begin{equation}
 \sum_{i = 1}^{r} \mu(c_i) \Big[\frac{W}{\sum_{i = 1}^{r} \upsilon(c_i)}\Big]
\label{label:we}
\end{equation}

Dividing the above by~\eqref{eq:EVarTemp5}, we get the following for
the ratio of the working energy to total energy:
\begin{equation}
\frac{\sum_{i = 1}^{r} \mu(c_i)}{\sum_{i = 1}^{r} (\mu(c_i) - 
\gamma(c_i)) + \Gamma }
\label{label:wt1}
\end{equation}

This in turn simplifies to:
\begin{equation}
 \frac{\sum_{i = 1}^{r} \mu(c_i)}
 {\sum_{i = 1}^{r} \mu(c_i) + \sum_{i = {r+1}}^{m} \gamma(c_i) }
\label{eq:wt2}
\end{equation}

Considering \eqref{eq:wt2} shows that the ratio of the working energy
to total energy can only be optimized by increasing the value of $r$,
i.e., when $r=m$, but this violates Theorem~\ref{th4} which tells us
that for minimum total energy consumption, the value of $r$ may not be
necessarily be equal to $m$.

Thus, a system running in such a way as to consume the
least possible energy while completing some work will not always be
most efficient in terms of the second measure, as the energy
consumption by the idle machines can be significant.  Conversely, if we
seek to optimize the system performance by the second measure and
reduce the idle consumption of machines, then the overall energy
consumption to do the work is not always minimized.

The data center energy efficiency metric called Power Usage
Effectiveness (PUE)~\cite{avelar2012}, an industry standard recommended
by the U.S. Environmental Protection Agency under its Energy Star
program,\footnote{See 
https://www.energystar.gov/ia/partners/prod\_development/downloads/DataCenterRat
ing\_General.pdf.}
is an energy efficiency measure of the second kind, as it considers
only the ratio of the total energy to the working energy.  It is thus
not a surprise that the PUE comes with its share of controversy and
criticisms~\cite{brady2013}.

\end{document}